\documentclass{sbc2023}%

\usepackage{graphicx}
\usepackage[misc,geometry]{ifsym} 
\usepackage{fontspec}
\usepackage{color}
\usepackage{hyperref} 
\usepackage{aas_macros}
\usepackage[bottom]{footmisc}
\usepackage{supertabular}
\usepackage{afterpage}
\usepackage{url}
\usepackage{pifont}
\usepackage{multicol}
\usepackage{multirow}
\usepackage{amsthm}
\usepackage{colortbl,xcolor}
\usepackage{nicematrix}
\usepackage[table]{xcolor}

\usepackage{newtxtext}
\usepackage{newtxmath}

\usepackage{algorithm}
\usepackage{algpseudocode}
\algrenewcommand\algorithmiccomment[1]{\quad\(\triangleright\) #1}
\algrenewcommand\algorithmiccomment[1]{\quad{\scriptsize$\triangleright$~#1}}

\usepackage{tkz-graph}
\usepackage{pgf, tikz}
\usetikzlibrary{graphs}
\usetikzlibrary{trees}
\usetikzlibrary{positioning}

\newif\ifarxiv
\arxivtrue 
\ifarxiv
  \usepackage[T1]{fontenc}
  \usepackage[utf8]{inputenc} 
  \usepackage{newtxtext}
  \usepackage{newtxmath}

  \providecommand{\faIcon}[1]{} 

\else
\fi

\newtheorem{theorem}{Theorem}
\newtheorem{defi}{Definition}
\newtheorem{lemma}{Lemma}
\newtheorem{example}{Example}
\newtheorem{proposition}{Proposition}

\newtheorem{corollary}{Corollary}

\setcitestyle{square}

\definecolor{orcidlogo}{rgb}{0.37,0.48,0.13}
\definecolor{unilogo}{rgb}{0.16, 0.26, 0.58}
\definecolor{maillogo}{rgb}{0.58, 0.16, 0.26}
\definecolor{darkblue}{rgb}{0.0,0.0,0.0}
\hypersetup{colorlinks,breaklinks,
            linkcolor=darkblue,urlcolor=darkblue,
            anchorcolor=darkblue,citecolor=darkblue}

\title[Analogy between List Coloring Problems and the Interval $k$-$(\gamma,\mu)$-choosability property: theoretical aspects of complexity]{Analogy between List Coloring Problems and the Interval $k$-$(\gamma,\mu)$-choosability property: theoretical aspects of complexity}

\author[Gama et al. 2025]{
\affil{\textbf{Simone Gama}~~[~\textbf{Universidade Federal do Amazonas}~|\href{mailto:simone.gama@icomp.ufam.edu.br}{~\textbf{\textit{simone.gama@icomp.ufam.edu.br}}}~]}

\affil{\textbf{Rosiane de Freitas}~~[~\textbf{Universidade Federal do Amazonas}~|\href{mailto:rosiane@icomp.ufam.edu.br}{~\textbf{\textit{rosiane@icomp.ufam.edu.br}}}~]}

}

\begin{document}

\begin{frontmatter}
\maketitle

\begin{abstract}
\textbf{Abstract.~}
\noindent
This work investigates structural and computational aspects of list-based graph coloring under
interval constraints. Building on the framework of analogous and p-analogous problems, we show
that classical List Coloring, $\mu$-coloring, and $(\gamma,\mu)$-coloring share strong
complexity-preserving correspondences on graph classes closed under pendant-vertex extensions.
These equivalences allow hardness and tractability results to transfer directly among the models. Motivated by applications in scheduling and resource allocation with bounded ranges, we introduce
the interval-restricted $k$-$(\gamma,\mu)$-coloring model, where each vertex receives an interval
of exactly $k$ consecutive admissible colors. We prove that, although $(\gamma,\mu)$-coloring is
NP-complete even on several well-structured graph classes, its $k$-restricted version becomes
polynomial-time solvable for any fixed $k$. Extending this formulation, we define
$k$-$(\gamma,\mu)$-choosability and analyze its expressive power and computational limits. Our
results show that the number of admissible list assignments is drastically reduced under interval
constraints, yielding a more tractable alternative to classical choosability, even though the
general decision problem remains located at high levels of the polynomial hierarchy. Overall, the paper provides a unified view of list-coloring variants through structural
reductions, establishes new complexity bounds for interval-based models, and highlights the
algorithmic advantages of imposing fixed-size consecutive color ranges.

\end{abstract}


\end{frontmatter}

\section{Introduction}\label{sec:intro}

Graphs are mathematical structures composed of a set of vertices $V$ and a set of edges $E$, 
where $G = (V,E)$. In the classical vertex coloring problem, one assigns colors to vertices so 
that adjacent vertices receive distinct colors. The minimum number of colors required for such 
an assignment is the chromatic number $\chi(G)$, and determining it is a central problem in 
graph theory, known to be NP-hard \citep{karp1972reducibility}. A $k$-coloring of $G$ induces a 
partition of $V$ into $k$ independent sets, and since all vertices in a clique are pairwise 
adjacent, they must receive different colors, implying $\chi(G) \ge \omega(G)$. Classical bounds 
relating $\chi(G)$ to structural parameters, such as clique number and maximum degree, are well 
studied \citep{brooks1941colouring,lovasz1975three}, and several works explore these relationships 
\citep{pardalos1998graph,lovasz1979shannon,Erdos1959GraphProbab,randerath2004vertex,mihokchromatic}.

Several constrained versions of coloring have been proposed. A fundamental generalization is 
\emph{List Coloring}, in which each vertex $v$ is assigned a list $L(v)$ of allowed colors, and 
the objective is to choose a proper coloring consistent with all lists. Introduced independently 
by Erd\H{o}s et al.\ in 1979 \citep{erdos:1979} and Vizing in 1976 \citep{vizing1976coloring}, 
List Coloring generalizes classical $k$-coloring, since the latter corresponds to the case 
$L(v)=\{1,\dots,k\}$ for all $v \in V(G)$. Due to this added generality and vertex-wise 
constraints, List Coloring is NP-complete in general graphs 
\citep{erdos:1979,vizing1976coloring,jansen1997generalized}.

The List coloring also has its variations, among them the Precoloring extension,  $\mu$-coloring and the $(\gamma,\mu)$-coloring:

\begin{itemize}
    \item \textbf{$\mu$-coloring}: This is a version of List coloring where each vertex may have different quantities of available colors, controlled by a parameter $\mu$. Introduced by Bonomo~\textit{et al.} \citep{bonomo2005between}, given a graph $G$ and a function $\mu: V \rightarrow \mathbb{N}$, $G$ is $\mu$-colorable if there exists a coloring $f$ of $G$ such that $f(v) \leq \mu(v)$ for every $v\in V$

    \item \textbf{$(\gamma,\mu)$-coloring}: An even more refined model that introduces two functions, $\gamma$ and $\mu$, to determine the restrictions on color assignment to vertices, allowing a more flexible approach to different types of problems.  Introduced by Bonomo~\textit{et al.} too \citep{bonomo2009exploring}, give a graph $G$ and a function $\gamma,\mu:V(G)\rightarrow \mathbb{N}$ such that $\gamma(v)\leq\mu(v)$ for every $v\in V(G)$, $G$ has a $(\gamma,\mu)$-coloring if there exists a coloring $f$ of $G$  such that $\gamma(v)\leq f(v)\leq\mu(v)$ for every $v\in V(G)$.

     \item \textbf{Precoloring extension}: In this variation, some vertices already have predefined colors, and the coloring must be extended to the rest of the graph while respecting this precoloring. Introduced by Tuza~\textit{et al.}~\citep{biro1992precoloring}, takes as input a graph \( G = (V, E) \), a subset \( W \subseteq V \), a coloring \( f' \) of \( W \), and a natural number \( k \). The goal is to decide whether \( G \) admits a \( k \)-coloring \( f \) such that \( f(v) = f'(v) \) for every \( v \in W \).
\end{itemize}

The List coloring problem also has some properties in relation to its color list, is when the color lists assigned to the vertices of $ G $ have size $ k $. A graph is $k$-choosable, or \textit{choice number} if it has a proper List coloring, no matter how one assigns a list of $k$ colors to each vertex. The choice number of a graph $G$, denoted by $\chi_{\ell}(G)$, is the least number $k$ such that $G$ is $k$-choosable. It is easy to see that the chromatic number $\chi_G$ is less than or equal to the choice number $\chi_{\ell}(G)$ (the smallest $k$ so that $G$ is $k$-choosable). Determining the choosability of a graph is a challenging problem, and for certain classes of graphs, many open questions remain.

Paul Erdős and his collaborators initiated studies on List coloring and choosability in graphs based on Dinitz’s Problem, introduced by Jeff Dinitz in 1978 \citep{dinitz1980lower, hofmannproofs, lastrina2012list}. This problem, originally formulated in terms of Latin squares, can be reinterpreted as a special case of list coloring in bipartite graphs, where each vertex has a specific set of allowed colors. Erdős, using probabilistic methods, demonstrated the existence of certain valid colorings under list constraints, making significant contributions to graph theory.

Despite its theoretical and practical relevance, List coloring and its variations remain an underexplored and challenging field, especially for specific graph classes. The difficulty in establishing exact bounds and the combinatorial complexity involved make this topic a fertile ground for new research and discoveries. Thus, investigating techniques and specific results for graph families can open new possibilities and reveal interesting structural properties in this domain.

\vspace{0.5cm}
\noindent\textbf{Application of List coloring.} The List coloring problem finds a natural application in processor assignment for operations in branching control flow graphs, as discussed by \cite{jansen1997generalized}.

In this setting, each operation is represented as a vertex in a graph, and available processors are seen as colors. However, not every processor can execute every operation. For each operation $v$, a list $S_v$ specifies which processors (colors) are capable of executing it, this defines the list of admissible colors for each vertex.

The graph's edges encode \textbf{incompatibilities}: two operations connected by an edge cannot be executed on the same processor. These incompatibilities typically arise from:
\begin{itemize}
    \item \textbf{Logical conflicts}, when operations belong to different branches of the control flow (modeled as a \textit{cograph}).
    \item \textbf{Temporal conflicts}, when operations overlap in their execution intervals (modeled as an \textit{interval graph}).
\end{itemize}

Often, the overall incompatibility graph is the intersection of a cograph and an interval graph. The objective is to find an assignment function $f: V \rightarrow S$ that satisfies:
\begin{enumerate}
    \item No two incompatible operations share the same processor ($f(v) \neq f(w)$ whenever $\{v,w\} \in E$).
    \item Each operation is assigned a processor from its list of allowed processors ($f(v) \in S_v$).
\end{enumerate}

This generalized List coloring model allows us to address realistic scheduling problems where both resource constraints and conflict relations must be respected. The approach from Jansen and Scheffler provides a theoretical foundation for solving such problems efficiently when the underlying conflict graphs have tree-like structure or special properties. An illustrative example of this application is shown in Figure~\ref{fig:applications-list-coloring}.

\begin{figure}[H]
    \centering
    \includegraphics[width=0.5\textwidth]{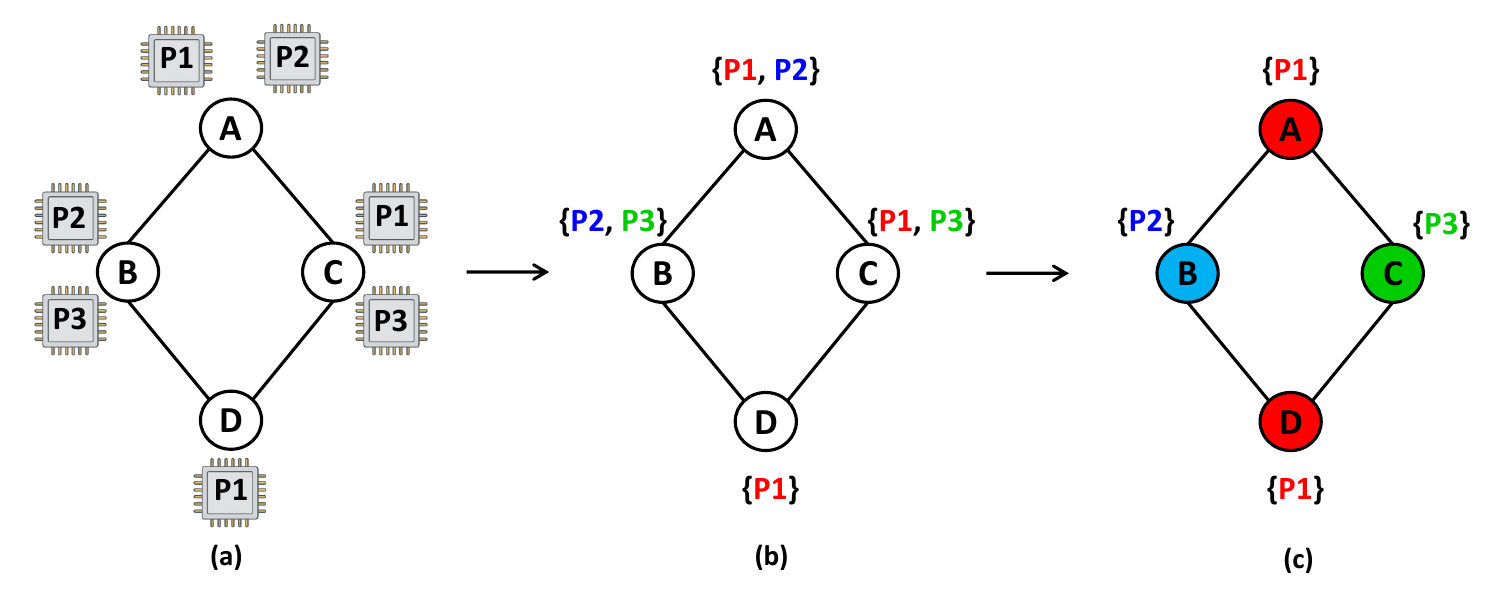}
    \caption{Example of list coloring applied to processor assignment. Each operation (vertex) has a list of allowed processors (colors), and the graph encodes execution conflicts. The final coloring respects both conflict constraints and individual compatibility restrictions.}
    \label{fig:applications-list-coloring}
\end{figure}

The Figure~\ref{fig:applications-list-coloring} illustrates the List Coloring problem (\textsc{LiCol}) as described by~\cite{jansen1997generalized}, in the context of processor assignment. Subfigure~(a) shows a set of operations represented as vertices, along with the available processors. Each vertex is associated with a subset of processors that are capable of executing the corresponding operation, reflecting the admissible color sets $S_v \subseteq S$. Subfigure~(b) represents the incompatibility graph, where edges denote logical or temporal conflicts that prevent two operations from sharing the same processor. Finally, subfigure~(c) presents a valid assignment (i.e., coloring) where each operation is mapped to a compatible processor from its list, while respecting the incompatibility constraints. This visual example satisfies both conditions defined in the \textsc{LiCol} formulation: conflicting vertices receive different colors, and each color assignment lies within the allowed processor list of each operation.

\vspace{0.5cm}
\noindent\textbf{Application of $\mu$-coloring.} The \citet{bonomo2005between}  introduce the \textit{$\mu$-coloring} model as a natural extension of classical coloring for resource allocation problems involving incompatibilities and individual requirements. In this setting, users compete for exclusive resources that can be ranked by quality or capacity, and each user requires a resource meeting a personal minimum threshold.

The system is modeled as an undirected graph \( G = (V, E) \), where vertices represent users and edges represent conflicts prohibiting shared resource usage. Each vertex \( v \) is assigned an upper bound \( \mu(v) \), indicating the highest resource level (or color) acceptable for that user.

A valid $\mu$-coloring is a function \( f: V \rightarrow \mathbb{N} \) satisfying:
\[
\begin{aligned}
f(v) &\neq f(w), && \forall \{v,w\} \in E, \\
f(v) &\leq \mu(v), && \forall v \in V.
\end{aligned}
\]

This model captures allocation situations where users accept any sufficiently good resource rather than a fixed list, making $\mu$-coloring less restrictive than list coloring. It is particularly suitable for scenarios like equipment or service assignments, where resources can be ordered by performance.

\begin{figure}[H]
    \centering
    \includegraphics[width=0.5\textwidth]{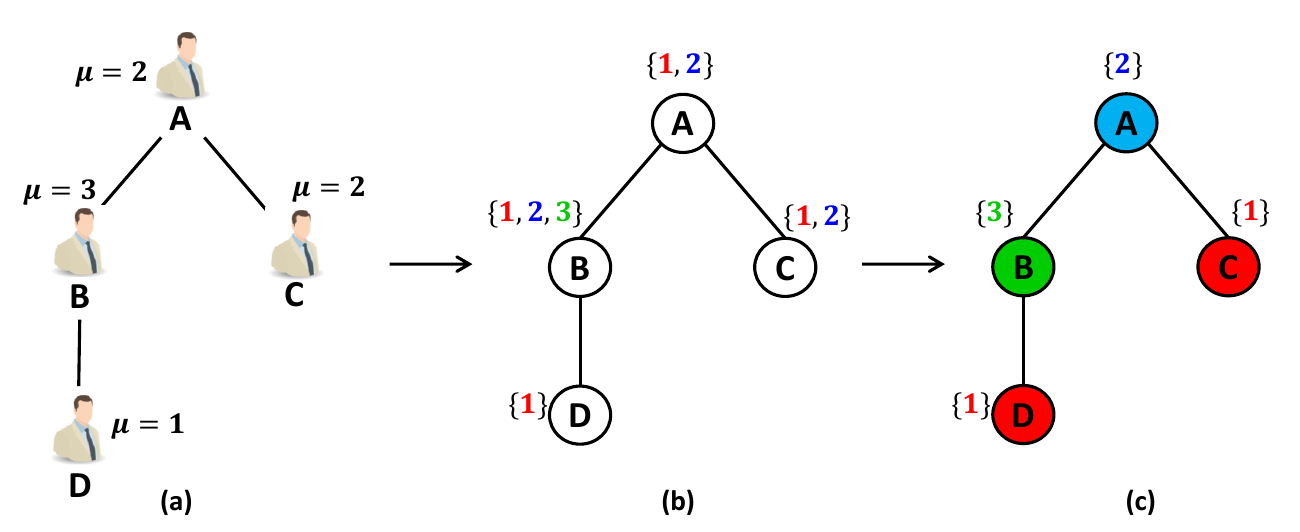}
    \caption{Illustration of the $\mu$-coloring model applied to a channel assignment scenario. Each vertex represents a client with a maximum acceptable channel index $\mu(v)$, and edges denote conflicts that require distinct channel assignments. The final coloring ensures that each client receives a suitable channel while avoiding interference with neighbors.}
    \label{fig:applications-mu-coloring}
\end{figure}

The Figure~\ref{fig:applications-mu-coloring} illustrates the $\mu$-coloring model as introduced by~\cite{bonomo2005between}, which extends classical vertex coloring to better represent resource allocation problems involving user-specific quality constraints. In this example, each vertex represents a corporate client requiring exclusive access to a communication channel, with an upper bound $\mu(v)$ indicating the highest channel index the client is willing to accept (lower values correspond to higher-quality resources). Subfigure~(a) shows the conflict graph among clients and their respective $\mu$ values. Subfigure~(b) presents the admissible color sets, defined as $\{1, 2, \ldots, \mu(v)\}$ for each vertex. Finally, subfigure~(c) displays a valid $\mu$-coloring where each client receives a channel that satisfies both its quality threshold and the exclusivity constraints imposed by the conflict graph. This concrete scenario matches the motivation described by Bonomo et al., where users can accept any resource \textit{good enough} making $\mu$-coloring a more natural and flexible model than list coloring in such settings.

The \citet{bonomo2005between} demonstrate that $\mu$-coloring lies between classical coloring and list coloring in both generality and complexity, and present polynomial-time algorithms for special graph classes such as cographs.

\vspace{0.5cm}
\noindent\textbf{Application of $(\gamma, \mu)$-coloring.} In embedded real-time systems, tasks must often be executed within specific time windows while avoiding resource conflicts. This scheduling scenario can be naturally modeled as a vertex coloring problem with interval constraints, precisely the setting captured by the $(\gamma, \mu)$-coloring model.

Each task is represented as a vertex in a graph, and an edge between two vertices indicates that the corresponding tasks cannot be executed simultaneously (e.g., due to sharing a non-preemptive resource such as a CPU core or communication bus). The execution time of a task is modeled as a color assigned to its vertex, subject to an interval constraint $f(v) \in [\gamma(v), \mu(v)]$, where $\gamma(v)$ and $\mu(v)$ represent the earliest and latest permissible execution times, respectively. The resulting function $f: V \rightarrow \mathbb{N}$ must be a proper coloring, such that $f(u) \ne f(v)$ for every edge $uv \in E$.

The Figure~\ref{fig:applications-gamma-mu-coloring} illustrates this modeling approach. Four tasks are shown, each with a light gray interval representing its feasible execution window $[\gamma(v), \mu(v)]$. The colored dot within each interval indicates the actual execution time $f(v)$ assigned to that task. Conflicts such as shared resources are modeled by ensuring that tasks A and B, as well as C and D, are scheduled at different times. This ensures a valid and conflict-free assignment of execution times within the system’s operational cycle.

This representation is both expressive and efficient for capturing key constraints in embedded system scheduling, especially when precise control over task timing and resource isolation is required.

\begin{figure}[H]
    \centering
    \includegraphics[width=0.5\textwidth]{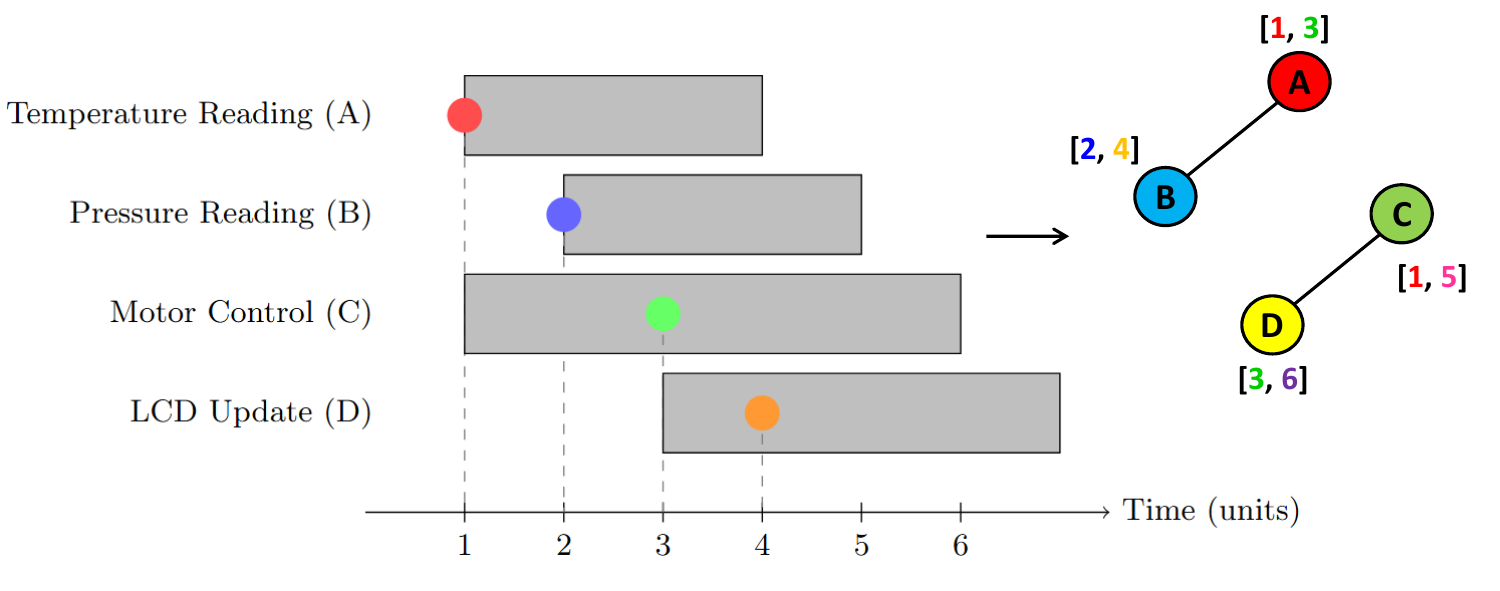}
    \caption{Modeling of a real-time task schedule as a $(\gamma, \mu)$-coloring problem. On the left, each horizontal bar represents the allowed execution window $[\gamma(v), \mu(v)]$ of a task, and the colored dot indicates the assigned time $f(v)$. On the right, each vertex in the conflict graph represents a task, with its color matching the scheduled time. Edges represent resource conflicts — i.e., tasks that cannot run at the same time. The interval $[\gamma(v), \mu(v)]$ for each task is also shown next to its corresponding vertex. This coloring ensures a valid and conflict-free schedule within the specified temporal constraints.}
    \label{fig:applications-gamma-mu-coloring}
\end{figure}

\vspace{0.5cm}
\noindent\textbf{Application of $k$-choosability.}
Consider the frequency assignment problem in a wireless network modeled by the path $P_3$ with vertices $A$--$B$--$C$. Each transmitter has a list of two available frequencies: $L(A)=\{100,102\}$, $L(B)=\{100,101,103\}$, and $L(C)=\{101,102\}$. The goal is to assign each transmitter a frequency from its list such that adjacent vertices use distinct frequencies. Choosing $103$ for $B$, we may assign $100$ to $A$ and $101$ to $C$, producing a valid interference-free assignment. More generally, since every vertex has a list of size~2, and $P_3$ is $2$-choosable, such a conflict-free assignment always exists regardless of the specific lists.

\paragraph{}
The computational complexity landscape of constrained coloring problems remains only partially 
understood, particularly regarding the relationships between models such as List Coloring, 
$\mu$-coloring, and $(\gamma,\mu)$-coloring. Although these variants are central to graph theory, 
it is still unclear when they are structurally equivalent, when their complexities diverge across 
graph classes, and under which conditions hardness or tractability can be transferred between them. 
This gap limits the development of a unified theoretical framework and complicates the reuse of 
existing algorithmic results.

\paragraph{}
At the same time, many real-world applications, including task scheduling with execution windows, 
resource allocation under bounded capacities, and frequency assignment, naturally require 
interval-based coloring constraints that specify lower and upper bounds for each vertex. 
However, $(\gamma,\mu)$-coloring remains computationally hard even in graph classes where 
classical coloring is easy, motivating the search for structured variants that reduce the 
combinatorial explosion inherent in arbitrary lists. These theoretical and practical needs justify 
the investigation of interval-restricted models such as $k$-$(\gamma,\mu)$-coloring and its 
choosability version, which aim to balance expressive power with improved algorithmic tractability.

\vspace{0.5cm}
\noindent\textbf{Our Results.} Building on the foundations established in Section~\ref{sec:preliminar}, where the classical 
framework of list assignments, $\mu$-coloring, and $(\gamma,\mu)$-coloring is presented, and on 
the contextual overview in Section~\ref{sec:trabalhos_relacionados}, this work develops a unified perspective 
on constrained graph coloring models. Our first main result, detailed in 
Section~\ref{sec:analogous-problems}, shows that \emph{List Coloring} and $(\gamma,\mu)$-coloring are 
analogous for graph classes closed under the pendant-vertex operator $\psi$. This structural 
equivalence implies that both problems share the same complexity behavior, allowing hardness and 
tractability results to transfer in both directions.

Section~\ref{sec:restricted-k-choosable} introduces the $k$-$(\gamma,\mu)$-coloring model, where each vertex is 
restricted to an interval of exactly $k$ consecutive colors. This structured restriction sharply 
reduces the combinatorial complexity of list assignments while preserving expressiveness. A key 
contribution is the demonstration that, although general $(\gamma,\mu)$-coloring remains 
NP-complete on bipartite, split, and interval graphs, its $k$-restricted version becomes 
polynomial-time solvable for fixed $k$ in all these classes, as summarized in Table~\ref{tab:coloring_complexity}. Moreover, 
a classical $k$-coloring can always be transformed efficiently into a $k$-$(\gamma,\mu)$-coloring, 
reinforcing the compatibility between the two models.

Finally, Section~\ref{subsubsec:k-gamma-mu-choosable} formalizes \textit{$k$-$(\gamma,\mu)$-choosability}, extending 
the interval framework to a choosability setting. Since every vertex receives an interval of size 
$k$, this model preserves the robustness of classical $k$-choosability while avoiding the 
exponential blowup of arbitrary lists, making it more suitable for applications involving bounded 
ranges or windowed resources.

Together, these contributions offer a clearer characterization of how structural constraints on 
color intervals influence the complexity and tractability of list-based coloring models.

\section{Preliminaries}\label{sec:preliminar}

Let $G=(V,E)$ be a simple graph, where $V$ is the set of vertices and $E$ is the set of edges. A graph $ G'=(V',E') $ is a subgraph of $ G $ if $ V' \subseteq V $ and $ E' \subseteq E $. A subgraph $ G'=(V',E') $ is an induced subgraph of $ G $ if $ E'= {uv:uv \in E \text{ and } u,v \in V'} $. We also say that $ G' $ is induced by $ V' $ and usually write $ G(V') $ for $ G' $. Let $ G $ be a graph and $ \{H_1,\ldots, H_p\} $ be a set of graphs. Then $ G $ is $ (H_1, . . . , H_p) $-free if $ G $ has no induced subgraph isomorphic to a graph in $ \{H_1,\ldots, H_p\} $.

The degree of a vertex in a graph is its number of incident edges. The degree of a graph $G$ (or its maximum degree) is the maximum of the degrees of its vertices, often denoted $\Delta(G)$. A sequence of vertices $ v_0, v_1, v_2,\ldots, v_l, v_0 $ is called a cycle of 	length $ l-1 $ (or closed path) if $ v_{i-1}v_i\in E $ for $ i=1,2,\ldots,l $ and $ v_lv_0\in E $. A graph $ G $ is bipartite if its vertices can be partitioned into two disjoint stable sets $ V = S_1 + S_2 $, i.e., every edge has one endpoint in $ S_1 $ and the other in $ S_2 $. An independent set, or a stable set, of $ G $ is a subset $ S $ of $ V $, where no two vertices of $ S $ are adjacent. We say that $ V' \subseteq V $ is a clique in $ G $ (or complete subgraph) if for all $ u,v\in V' $, $ u \neq v $, $ uv \in E $. We can define a clique number (or maximum clique) by $ \omega(G) $.

For a graph $G=(V,E)$, an assignment $ c:V \rightarrow N $ is a coloring of $ G $. Furthermore, this coloring is proper if $ c(u)\neq c(v) $ for all $ uv\in E $, that is, a \textit{$k$-coloring} of $G$ is an assignment of colors to the vertices of $G$ such that no two adjacent vertices share the same color. The chromatic number $\chi_G$ of a graph is the minimum value of $k$ for which $G$ is $k$-colorable.

A $ k $-coloring naturally induces a partition of $ V $ into $ k $ color classes such that the members of each class are assigned the same color, i.e., they are pairwise non-adjacent. Therefore, a partition of $ V $ into $ k $ independent sets is equivalent to a $ k $-coloring of $ G $. Since all vertices of a clique number $ G $ are pairwise adjacent, it is immediate to see that they must have different colors in any coloring of $ G $, then $ \chi(G)\geq\omega(G) $ \citep{korman1979graph}. There are other bounds in $ \chi(G) $ involving clique, for example, if $ G $ is not a clique or an odd cycle, then $ \chi(G)\leq\Delta $ \citep{brooks1941colouring,lovasz1975three}.

The classic graph coloring problem, which consists in finding the chromatic number of a graph, is one of the most important combinatorial optimization problems and it is known to be NP-complete \citep{karp1972reducibility}. There are several versions of this classic vertex coloring problem, involving additional constraints on both the edges and the vertices of the graph. One of them is the List coloring problem, where given a graph $G$ there is an associated set $L(v)$ of allowed color lists for each vertex $v \in V(G)$. If it is possible to get a proper coloring of $G$ with these color lists, then we have a list coloring of $G$. A \textit{list-assignment} $L$ to the vertices of $G$ is the assignment of a list (set) $L(v)$ of colors to every vertex $v$ of $G$ (\textbf{Figure}~\ref{fig:list_coloring}).

\begin{defi}[\citep{erdos:1979}]
	If it is possible to get a proper coloring of $G$ with set $L(v)$ of colors, then we have a List coloring of $G$.
\end{defi}

\begin{figure}[h]
\centering
\includegraphics[width=0.5\textwidth]{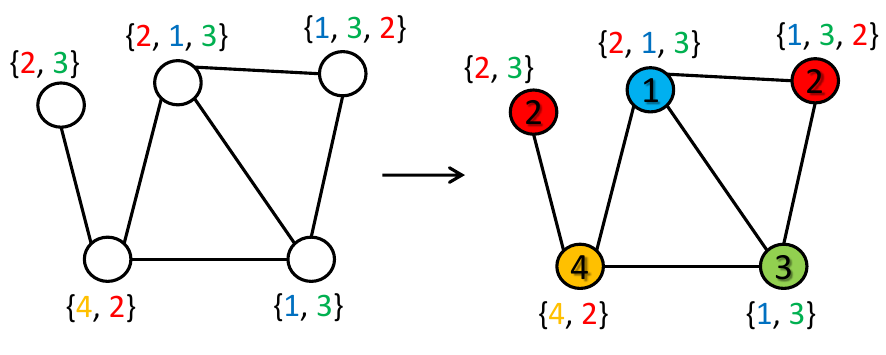}

\caption{Example of a list coloring on a graph. On the left, each vertex is associated with a list of allowed colors $L(v)$, shown next to the vertex. On the right, a proper list coloring is shown, where each vertex $v$ is assigned a color $f(v) \in L(v)$ such that no two adjacent vertices share the same color. The assigned colors are consistent with the lists and respect the coloring constraint for adjacent vertices.}
\label{fig:list_coloring}
\end{figure}

List coloring generalizes proper coloring by allowing each vertex to have its own specific list of permitted colors, from which a color must be chosen such that adjacent vertices receive different colors. From this model, important variations arise, such as Precoloring, $\mu$-coloring, and $(\gamma, \mu)$-coloring.

Precoloring extension \cite{biro1992precoloring} can be seen as a special case of List coloring, which refers to the problem of coloring a graph where no vertices have been precolored, but each vertex is assigned a list of available colors. To reduce a precoloring extension problem to a list coloring problem, assign each uncolored vertex a list consisting of the colors not yet used by its initially precolored neighbors, and then remove the precolored vertices from the graph. In other words, a prespecified vertex subset is colored beforehand, and the goal is to extend this partial proper coloring to a proper $k$-coloring of the whole graph. This type of constraint, where some vertices have restricted coloring options due to prior assignments and illustrates a broader class of coloring problems that incorporate individual limitations on vertex colors.

Among the most studied of these constrained coloring variants is the $\mu$-coloring. In this version, each vertex $v$ in the graph $G$ is associated with a natural number $\mu(v)$, which defines the maximum allowable color that can be assigned to it. A proper $\mu$-coloring is a coloring function $f$ such that $f(v) \leq \mu(v)$ for every $v \in V(G)$. This constraint can model situations where certain vertices are limited in their color choices due to physical, resource, or priority limitations. 

\begin{defi}[\citep{bonomo2005between}]
	Given a graph $G$ and a function $\mu:V(G)\rightarrow \mathbb{N}$, $ G $ is $\mu$-colorable if there exists a coloring $f$ of $G$ such that $f(v)\leq\mu(v)$ for every $v\in V(G)$.
\end{defi}

An even more general and flexible version is the $(\gamma, \mu)$-coloring, where each vertex $v$ is associated with both a lower bound $\gamma(v)$ and an upper bound $\mu(v)$ for its color. Formally, given two functions $\gamma, \mu : V(G) \rightarrow \mathbb{N}$ such that $\gamma(v) \leq \mu(v)$ for all $v \in V(G)$, we say that the graph $G$ is $(\gamma, \mu)$-colorable if there exists a proper coloring $f : V(G) \rightarrow \mathbb{N}$ satisfying $\gamma(v) \leq f(v) \leq \mu(v)$ for every vertex $v$. This model is useful in scenarios where a vertex not only has a restriction on how large the assigned value can be, but also a minimum threshold it must respect. As such, $(\gamma, \mu)$-coloring generalizes both the classical list coloring and the $\mu$-coloring, and is widely studied in the context of constraint satisfaction, resource allocation, and scheduling problems. An illustrative example of a $(\gamma, \mu)$-coloring, including both the assigned color $c(v)$ and the corresponding bounds $(\gamma(v), \mu(v))$ for each vertex, is shown in \textbf{Figure~\ref{fig:gamma-mu-coloring}}.

\begin{defi}[\citep{bonomo2009exploring}]
	Give a graph $G$ and a function $\gamma,\mu:V(G)\rightarrow \mathbb{N}$ such that $\gamma(v)\leq\mu(v)$ for every $v\in V(G)$, $G$ is $(\gamma,\mu)$-colorable if there exists a coloring $f$ of $G$  such that $\gamma(v)\leq f(v)\leq\mu(v)$ for every $v\in V(G)$.
\end{defi}

\begin{figure}[h!]
\centering
\includegraphics[width=0.5\textwidth]{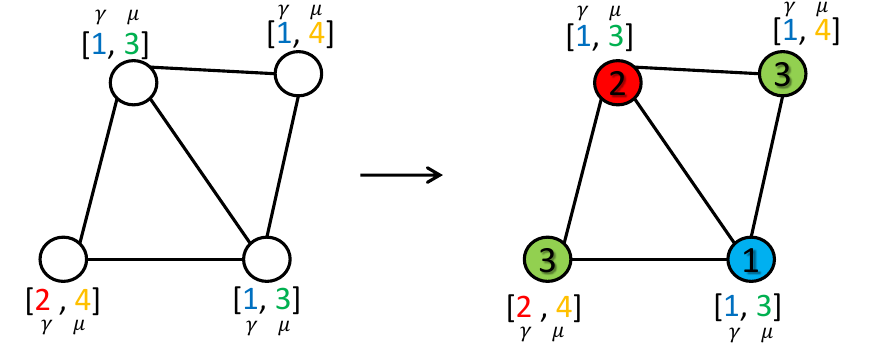}
\caption{An example of a $(\gamma, \mu)$-coloring of a graph. On the left, each vertex is associated with an interval $[\gamma(v), \mu(v)]$ indicating the minimum and maximum allowed colors for that vertex. On the right, a valid coloring is presented in which each vertex $v$ is assigned a color $f(v)$ such that $\gamma(v) \leq f(v) \leq \mu(v)$ and adjacent vertices receive different colors. The vertex colors are shown numerically and visually.}

\label{fig:gamma-mu-coloring}
\end{figure}

In addition to the upper bounding function $\mu:V(G) \rightarrow \mathbb{N}$, this generalization of the $\mu$-coloring problem also incorporates a function $\gamma:V(G) \rightarrow \mathbb{N}$, which defines lower bounds for the colors assigned to the vertices of the graph $G$. Thus, the coloring must simultaneously satisfy both lower and upper constraints at each vertex. This model represents a natural extension of the precoloring extension problem, encompassing it as a particular case and broadening its applicability to contexts involving more complex constraints.

\subsection{$k$-choosability in graphs}

The problem of the \textit{$k$-choosability} is a fundamental property within the framework of List coloring, a model that generalizes classical proper vertex coloring. Instead of sharing a common color set, each vertex \( v \) in a list coloring is assigned its own list of allowed colors, denoted by \( L(v) \). A list coloring of a graph \( G \) is a function that assigns to each vertex a color from its list such that adjacent vertices receive different colors, that is, the coloring remains proper. A graph is said to be \textit{$k$-choosable} if, for every possible assignment of lists of size \( k \) to the vertices, there exists a proper list coloring. The smallest such \( k \) is known as the graph’s \textit{choice number}. 

\begin{defi}[\citep{chartrand2019chromatic}]
A graph $G$ is \emph{$k$-choosable} if it is $\mathcal{L}$-list-colorable for every collection $\mathcal{L}$ of lists $L(v)$ assigned to the vertices of $G$, such that $|L(v)| \leq k$ for every $v \in V(G)$.
\end{defi}

The concept of $k$-choosability was independently introduced by Erdős and by Vizing, who conjectured that every graph satisfies \( \chi_\ell(G) \leq \Delta(G) + 1 \), where \( \chi_\ell(G) \) is the choice number and \( \Delta(G) \) is the maximum degree of \( G \).

Thus, a graph is said to be $k$-choosable (or to have choice number $k$) if it has a proper List coloring for every possible assignment of lists of $k$ colors to each vertex. The choice number of a graph $G$, denoted by $\chi_{\ell}(G)$, is the least number $k$ such that $G$ is $k$-choosable. It is easy to see that the chromatic number $\chi(G)$ is less than or equal to the choice number $\chi_{\ell}(G)$ (i.e., the smallest $k$ such that $G$ is $k$-choosable).

\begin{defi}[\citep{chartrand2019chromatic}]
The \emph{list chromatic number} (also known as the \emph{choice number}) $\chi_\ell(G)$ of a graph $G$ is the smallest positive integer $k$ such that $G$ is $k$-choosable.
\end{defi}

The Figure~\ref{fig:P3-choosable} shows that the path graph \( P_3 \) is $2$-choosable. Regardless of how lists of size 2 are assigned to its vertices, a proper list coloring is always possible.

\begin{figure}[h!]
\centering
\includegraphics[width=0.5\textwidth]{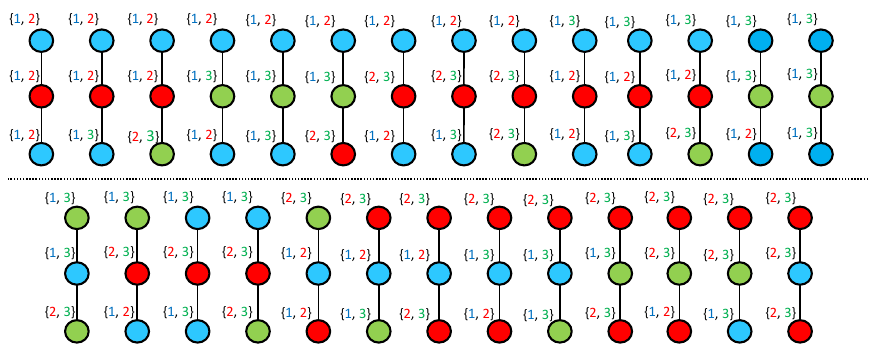}
\caption{Exhaustive verification that the path graph $P_3$ is $2$-choosable. Each subfigure corresponds to one of the $3^3 = 27$ possible list assignments where every vertex receives a list of two colors selected from the set $\{1,2,3\}$. For each assignment, a valid list coloring is shown, demonstrating that a proper coloring can always be selected from the assigned lists. The figure thus confirms that $P_3$ is indeed $2$-choosable.}

\label{fig:P3-choosable}
\end{figure}

To show that the path graph \( P_3 \) is 2-choosable, we consider all possible assignments of lists of size 2 to its three vertices. Using the color set \(\{1,2,3\}\), which generates all distinct 2-element lists (\(\{1,2\}, \{1,3\}, \{2,3\}\)), there are \(3^3 = 27\) total combinations. For each assignment, it is always possible to choose a color for each vertex from its list such that adjacent vertices receive different colors. This can be verified by case analysis or algorithmically, confirming that \( P_3 \) is indeed 2-choosable.

\subsection{Parameterized Complexity}

A parameterized problem $\Pi$ is informally defined by an instance, a parameter, and a question. 
A parameterized problem $\Pi(k)$ is \textit{fixed-parameter tractable} ($FPT$) if it can be solved 
in time $f(k)\cdot n^{c}$, where $n$ is the size of the input, $c$ is a constant, and $f$ is a 
computable function. The class $XP$ contains all problems solvable in time $\|I\|^{f(p)}$. 
Parameterized problems are further organized into the $W$-hierarchy: a problem belongs to $W[t]$ 
if it can be reduced, in $FPT$-time, to a weight-$k$ satisfiability problem on circuits of 
weft~$t$. By definition, $FPT \subseteq XP$, and the hierarchy extends through several 
intermediate classes:

\begin{equation}\label{Eq:ParameteComple}
FPT = W[0] \subseteq W[1] \subseteq W[2] \subseteq \ldots \subseteq W[P] \subseteq XP.
\end{equation}

Downey and Fellows pioneered the development of parameterized complexity theory, which refines 
the analysis of NP-complete problems by incorporating a numerical parameter $k$, not necessarily 
dependent on input size. Comprehensive references on the topic include 
\cite{abrahamson1989complexity,cai1997fixed,downey1998parameterized,Downey1999ParameterizedC,niedermeier2002invitation}.

One of the earliest results connecting parameterized complexity and graph coloring is due to 
Arnborg and Proskurowski~\cite{arnborg1989linear}, who proved that classical coloring is $FPT$ 
when parameterized by the treewidth. Surveys and advances on parameterized coloring can be found 
in \cite{arnborg1989linear,golovach2017survey,couturier2012parameterized}. Regarding 
\textsc{List Coloring}, Jansen and Scheffler~\cite{jansen1997generalized} showed that the problem 
is $FPT$ on $P_4$-free graphs (cographs) when parameterized by the size $k$ of the list 
assignment. Further results by Couturier et al.~\cite{couturier2012parameterized} show that 
\textsc{List Coloring} is $FPT$ on $(rP_1+P_2)$-free graphs with parameter $k+r$, and also 
$FPT$ on $(P_1+P_3)$-free graphs when parameterized solely by $k$.

The parameterized hardness of \textsc{List Coloring} is well established. Fellows et al.\ 
\cite{fellows2011complexity} proved that it is $W[1]$-hard when parameterized by treewidth, via a 
parameter-preserving reduction from \textsc{Multicolor Clique}, which is $W[1]$-hard 
\cite{fellows2007fixed}. Hardness persists even under structural restrictions: Fiala et al.\ 
\cite{fiala2011parameterized} showed that \textsc{List Coloring} remains $W[1]$-hard on 
split graphs, whose vertices can be partitioned into a clique and an independent set. These 
results collectively demonstrate that the parameterized intractability of \textsc{List Coloring} 
is robust across various restricted graph classes.

\section{Related Work}\label{sec:trabalhos_relacionados}

The main results in the literature on List coloring of graphs include the characterization of graph classes where list coloring behaves similarly to traditional coloring, such as trees and odd-length cycles. A significant breakthrough was Alon's \citep{alon1993restricted, alon2000degrees} use of the probabilistic method to establish upper bounds on the number of colors required to ensure a valid list coloring. Additionally, structural results show that certain graphs, such as planar graphs, admit list colorings with a controlled number of colors, contributing to the understanding of the differences between classical coloring and List coloring.

List coloring in some graph classes can be found in polynomial time, as is the case for tree graphs \citep{jansen1997generalized}, complete graphs~\citep{song2013graph}, and block graphs~\citep{bonomo2009exploring}.

 The List coloring for some graph classes is $ NP $-Complete, such as bipartite~\citep{kubale1992some}, complete bipartite~\citep{jansen1997generalized}, complement of a bipartite~\citep{jansen1997optimum}, cographs~\citep{jansen1997generalized}, distance-hereditary~\citep{jansen1997generalized}, split~\citep{bonomo2009exploring}, complete split~\citep{jansen1997generalized}, chordal~\citep{song2013graph}, interval~\citep{song2013graph,bonomo2009exploring} and line of $ K_n $~\citep{kubale1992some}.

It can be observed that the vertex coloring problem is a special case of $\mu$-coloring and the Precoloring extension, both of which are, in turn, special cases of $(\gamma, \mu)$-coloring. Moreover, $(\gamma, \mu)$-coloring is itself a particular case of List coloring. These observations imply that all problems within this hierarchy are polynomially solvable in those graph classes where list coloring is polynomial. Conversely, all of these problems are NP-complete in the graph classes where vertex coloring is NP-complete \cite{bonomo2009exploring}.

The \textbf{Table~\ref{tab:coloring-complexity}} summarizes the computational complexity of various vertex coloring problems including $\mu$-coloring, $(\gamma, \mu)$-coloring, and List coloring, across several graph classes. A key observation is the hierarchical relationship among these problems: vertex coloring is a special case of Precoloring, which in turn is a special case of $\mu$-coloring, itself a subset of $(\gamma, \mu)$-coloring, and all are particular instances of list coloring.

For general graphs, all these coloring problems are NP-complete (NP-C, for short), confirming the expected computational hardness. However, for structured graph classes such as trees, cographs, and interval graphs, the problems become polynomial-time solvable (P, for short), highlighting the influence of graph topology on complexity. For instance, all problems are solvable in polynomial time on trees and cographs, but even small increases in structural complexity (e.g., clique-tree height) can cause a shift to NP-completeness.

Another notable trend is the consistent NP-completeness of List coloring in most graph classes, aligning with its status as a generalization of other coloring problems. The table also illustrates cases where lower-complexity problems remain polynomial, while more general ones become intractable, emphasizing the nuanced transitions within the hierarchy.

\begin{table*}[ht]
\caption{
Computational complexity results for different vertex coloring problems across various 
graph classes. The columns correspond to $\mu$-coloring, $(\gamma, \mu)$-coloring, 
and List Coloring problems, while the rows represent distinct graph classes. 
Entries marked with ``P'' indicate polynomial-time solvability, ``NP-C'' denotes 
NP-completeness, and ``?'' indicates that the computational complexity remains open.
}
\centering\small
\label{tab:coloring-complexity}

\resizebox{\textwidth}{!}{%
\begin{tabular}{
|p{2.93cm}
|>{\centering\arraybackslash}p{4.3cm}
|>{\centering\arraybackslash}p{4.0cm}
|>{\centering\arraybackslash}p{4.5cm}|
}
\hline\hline
\textbf{Graph Class} & \textbf{$\mu$-coloring} & \textbf{$(\gamma,\mu)$-coloring} & \textbf{List coloring} \\ 
\hline

General                 & \footnotesize NP-C \citep{bonomo2005between} & \footnotesize NP-C \citep{bonomo2009exploring} & NP-C \citep{gravier1996coloration} \\
\hline

Block                   & P & P & P \citep{bonomo2009exploring} \\
\hline

Bipartite               & \footnotesize NP-C \citep{bonomo2005between} & NP-C \citep{bonomo2009exploring} & NP-C \citep{kubale1992some} \\
\hline

Clique-Tree height 1    & P \citep{bonomo2012coloring} & P \citep{bonomo2012coloring} & NP-C \\
\hline

Clique-Tree height 2    & P \citep{bonomo2012coloring} & NP-C \citep{bonomo2012coloring} & NP-C \\
\hline

Clique-Tree height 3    & NP-C \citep{bonomo2012coloring} & NP-C \citep{bonomo2012coloring} & NP-C \\
\hline

Cographs                & P \citep{bonomo2005between} & ? & NP-C \citep{jansen1997generalized} \\
\hline

Complete Bipartite      & P \citep{bonomo2009exploring} & P \citep{bonomo2009exploring} & NP-C \citep{jansen1997generalized} \\
\hline

Complete Split          & P \citep{bonomo2009exploring} & P \citep{bonomo2009exploring} & NP-C \citep{jansen1997generalized} \\
\hline

Distance Hereditary     & NP-C \citep{bonomo2009exploring} & NP-C \citep{bonomo2009exploring} & NP-C \citep{jansen1997generalized} \\
\hline

Interval                & NP-C \citep{bonomo2009exploring} & NP-C & NP-C \\
\hline

Line of $K_{n,n}$       & NP-C \citep{bonomo2009exploring} & NP-C & NP-C \\
\hline

Line of $K_n$           & NP-C \citep{bonomo2009exploring} & NP-C \citep{bonomo2009exploring} & NP-C \citep{kubale1992some} \\
\hline

Split                   & NP-C \citep{bonomo2009exploring} & NP-C \citep{bonomo2009exploring} & NP-C \\
\hline

Trees                   & P & P & P \citep{jansen1997generalized} \\
\hline

Unit Interval           & NP-C \citep{bonomo2012coloring} & NP-C & NP-C \\
\hline\hline
\end{tabular}
}
\end{table*}

The table summarizes the computational complexity of three vertex coloring problems, $\mu$-coloring, $(\gamma, \mu)$-coloring, and list coloring, across various graph classes. These classes exhibit well-established structural relationships that help explain the observed complexity patterns.

Many of the classes form a strict inclusion hierarchy, such as:
\begin{itemize}
    \item \textbf{Block graphs} $\subset$ \textbf{Chordal} $\cap$ \textbf{Distance-Hereditary}
  \item \textbf{Cographs} $\subset$ \textbf{Distance-Hereditary} $\subset$ \textbf{General} (\cite{ahn2022towards})
  \item \textbf{Complete Split} $\subset$ \textbf{Split} 
  \item \textbf{Clique-Tree height 1} $\subset$ height 2 $\subset$ height 3 (\cite{bonomo2012coloring})
  \item \textbf{Trees} $\subset$ \textbf{Block graphs}
  \item \textbf{Unit Interval} $\subset$ \textbf{Interval} (\cite{bonomo2009exploring})
\end{itemize}

As a general trend, the more structurally restricted the graph class, the more likely the coloring problems are solvable in polynomial time. For instance, all three problems are in P for \emph{trees}, \emph{block graphs}, and \emph{clique-tree graphs of height 1}. In contrast, as the structural constraints are relaxed (e.g., moving to height 2 or 3) the problems often become NP-complete.

An interesting exception is the case of \textbf{cographs}, for which:
\begin{itemize}
  \item $\mu$-coloring is in P,
  \item List coloring is NP-complete,
  \item The complexity of $(\gamma, \mu)$-coloring is still \textbf{open}.
\end{itemize}

List coloring tends to remain NP-complete even in restricted classes such as \emph{cographs}, \emph{interval}, \emph{split}, and \emph{bipartite} graphs. This is due to the stronger constraints imposed by the list assignment, which significantly increases problem hardness. Overall, the structural hierarchy not only organizes the graph classes but also serves as a useful framework to anticipate algorithmic tractability or intractability across different variants of coloring.

\subsection{Related work in choosablity in graphs}\label{subsec:work-choosability}

These complexity patterns also motivate further investigation into more general frameworks for graph coloring, among which the concept of \textit{choosability}, or list choosability, plays a central role. Choosability extends the traditional List coloring model by asking whether, for any assignment of lists of a given size to each vertex, there always exists a proper coloring from those lists.

Choosability provides a more robust framework for studying coloring under constraints, particularly in contexts where fixed color availability per vertex reflects real-world limitations, such as frequency assignment, scheduling, or resource distribution. While every $k$-colorable graph is trivially $k$-choosable for some cases, the converse does not hold in general, and determining choosability often involves deeper combinatorial arguments and complexity considerations.

This concept not only generalizes classical coloring but also connects deeply with structural graph theory, and its complexity varies significantly depending on the graph class under consideration. As such, choosability has become a central topic in the study of graph coloring and constraint satisfaction.

Choosability in graphs was introduced by Erdős \textit{et al.} in 1979~\citep{erdos:1979}, and independently by Vizing in 1976~\citep{vizing1976coloring}. Erdős \textit{et al.} characterized graphs that are 2-choosable, showed that bipartite graphs are 3-choosable, and conjectured that planar graphs are 5-choosable and that there exist planar graphs that are not 4-choosable. These conjectures were later proven by~\cite{thomassen1994every} and~\cite{voigt1993list}. 

Choosability theory has generated several extensions that impose additional structure on lists or restrict the admissible interactions between assigned colors. These generalized frameworks capture finer coloring requirements, encode specific structural constraints, and enforce combinatorial conditions absent from the classical model. They expose behaviors not exhibited by standard list coloring and substantially broaden the analytical scope of choosability. Below we outline some of the principal such extensions.

\begin{itemize}
    \item \textbf{(a : b)-choosability (or list multicoloring).}
    Introduced by~\cite{erdos:1979} and later developed by~\cite{gutner2009some}.  A graph $G$ is $(a:b)$-choosable if, for every list assignment $L$ with $|L(v)| = a$ 
    for each vertex $v$, there exists an assignment selecting exactly $b$ colors from 
    $L(v)$ for each vertex, such that adjacent vertices receive disjoint color sets.  
    Two fundamental results include:  
    (i) $(a:b)$-choosability is strictly stronger than $a/b$-colorability (Gutner–Tarsi);  
    (ii) $(a:b)$-choosability does not necessarily imply $(c:d)$-choosability even when 
    $c/d > a/b$ (Gutner–Tarsi).

    \item \textbf{Defective choosability.}
    Introduced by Cowen, Goddard, and Jesurum \citep{cowen1997defective} and later extended to list assignments by several authors.  
    A graph $G$ is $t$-defective $k$-choosable if, for every list assignment $L$ with $|L(v)| \ge k$ for each vertex $v$, there exists a list-coloring in which each color class induces a subgraph of maximum degree at most $t$.  
    Two representative results are:  
    (i) planar graphs of sufficiently large girth are $t$-defective $k$-choosable for small 
    values of $t$ and $k$ (Cowen–Goddard–Jesurum; Havet–Sereni);  
    (ii) for $t \geq 1$, many graph classes that are not $k$-choosable become 
    $t$-defective $k$-choosable (Havet–Sereni).

    \item \textbf{DP-choosability.}
    Introduced by \cite{dvorak2015correspondence}.  A graph $G$ is DP-$k$-choosable if, for every correspondence-cover (DP-cover) 
    $(L,H)$ with $|L(v)| = k$ for each vertex, there exists an independent transversal 
    of $H$ selecting one element from each list $L(v)$.  
    Key results include:  
    (i) DP-choosability strictly generalizes classical list choosability (Dvořák–Postle);  
    (ii) every planar graph is DP-5-choosable, strengthening Thomassen's 5-choosability 
    theorem (Dvořák–Postle; Bernshteyn).

    \item \textbf{Online choosability (or paintability).}
    Introduced by \cite{schauz2009mr} and based on earlier work by \cite{zhu2009line} on online list coloring.  
    A graph $G$ is online $k$-choosable if the colorer has a winning strategy in the adversarial list-coloring game against any sequence of list revelations of size at least $k$.  
    Two central results are:  
    (i) the online choice number is always at least the classical choice number (Schauz);  
    (ii) the online and offline choice numbers differ for several classes of graphs, 
    including bipartite and planar families (Schauz; Zhu).
\end{itemize}

To contextualize the position of our restricted form of 
$k$-choosability within the broader landscape of list-coloring theory, we summarize below the computational complexity of several major generalizations of choosability. These variants illustrate how different structural or algorithmic constraints affect the difficulty of deciding whether a graph satisfies a given list-coloring property. Table~\ref{tab:choosability-generalizations} provides a concise comparison of the main frameworks, their known complexity status, and representative references.

In all four generalizations considered, the structural constraints imposed on the 
lists differ fundamentally from those in interval-based models such as 
$(\gamma,\mu)$-coloring. In $(a:b)$-choosability, each vertex receives an arbitrary 
list of fixed size $a$, and the generalization arises not from additional structure 
on the lists but from requiring the selection of $b$ colors per vertex. In defective 
choosability, the lists remain completely arbitrary, with the restriction placed 
instead on the maximum defect allowed within each color class. In DP-choosability, 
the lists also have no internal structure; however, the interaction between lists 
is governed by an auxiliary correspondence graph that specifies arbitrary matchings 
between colors of adjacent vertices. Online choosability likewise employs arbitrary 
lists, but they are revealed adversarially and incrementally during an online 
coloring game. Thus, unlike interval-based frameworks, these generalizations 
preserve full arbitrariness in the formation of $L(v)$, shifting the complexity 
to constraints on color interactions or on the dynamics of the coloring process.

\begin{table*}[ht]
\centering
\caption{
Computational complexity landscape for major extensions of graph choosability. 
Each generalization imposes additional structural or algorithmic constraints on list assignments, 
resulting in decision problems of the form “does $G$ satisfy the corresponding list-coloring 
property for fixed parameters?”. The table summarizes known worst-case complexity classifications, 
highlighting that most variants lie at the second level of the polynomial hierarchy 
($\Pi_2^P$-complete), while online choosability naturally fits into the PSPACE framework and 
remains without a complete hardness characterization.
}
\label{tab:choosability-generalizations}

\begin{tabular}{|p{2.8cm}|p{4.2cm}|p{5.5cm}|p{3cm}|}
\hline\hline
\textbf{Generalization} & 
\textbf{Computational Complexity} &
\textbf{Brief Explanation} &
\textbf{References} \\ 
\hline

(a:b)-choosability & 
$\Pi_2^P$-complete (for fixed $a,b$) &
Decision requires verifying that \emph{every} list assignment of size $a$ admits a $b$-set coloring; includes $k$-choosability as special case. &
\cite{erdos:1979}; \cite{gutner2009some} \\
\hline

Defective $k$-choosability &
$\Pi_2^P$-complete (for fixed $k,d$) &
Generalizes list coloring by allowing each color class to have bounded defect; contains classical choosability when $d=0$. &
\cite{cowen1997defective}; \cite{havet2009improper} \\
\hline

DP-choosability (correspondence choosability) &
$\Pi_2^P$-complete (choosability); 
NP-complete (DP-$k$-colorability for fixed cover) &
Strict generalization of list coloring with arbitrary correspondences between lists; testing DP-$k$-colorability generalizes testing $k$-colorability. &
\cite{dvorak2015correspondence}; \cite{bernshteyn2017dp} \\
\hline

Online choosability (paintability) &
In PSPACE (decision problem);
no full complexity classification known &
Defined via an adversarial game; choice number is a lower bound for the paint number; complexity classification largely open. &
\cite{schauz2009mr}; \cite{zhu2009line} \\
\hline\hline
\end{tabular}

\end{table*}

\section{A note about Analogous Problems in List Coloring}\label{sec:analogous-problems}

The notion of analogous problems was introduced by Fellows et al.~\cite{fellows2015tractability} 
as a framework for comparing the structural complexity of distinct computational problems. 
Among their results, they proved that the Shortest Common Supersequence (SCS) problem and its 
restricted version RSCS are analogous, illustrating the applicability of the framework beyond 
graph-theoretic settings. The work~\cite{gama2019aspects} presented several observations 
indicating that certain specific variants of List Coloring exhibit strong structural similarities. 

Building on these insights, we now develop a set of results that not only formalize such 
relationships but also refine and extend them, providing strengthened versions of previously 
known correspondences among coloring problems. To state these results rigorously, we first recall 
the definitions of analogous and $p$-analogous problems, which serve as the formal basis for 
transferring complexity properties across different coloring formulations.

Denote by $Y(\Pi)$ the set of all instances $I$ of $\Pi$ yielding a yes-answer for the question ``$I\in Y(\Pi)?$''. The notion of {\em analogous problems} was introduced by \cite{fellows2015tractability}. Two optimization problems $\Pi$ and $\Pi'$ are said to be analogous if there exist linear-time reductions $f,g$ such that:

\begin{itemize}
    \item $\Pi \propto^{f} \Pi'$ and $\Pi' \propto^{g} \Pi$;
    \item every feasible solution $s$ for an instance $I$ of $\Pi$ implies a feasible solution $s'$ for $f(I)$ such that $size(s) = size(s')$;
    \item every feasible solution $s'$ for an instance $I'$ of $\Pi'$ implies a feasible solution $s$ for $g(I')$ such that $size(s') = size(s)$.
\end{itemize}

The definition above formalizes when two optimization problems exhibit equivalent solution structures under linear-time reductions that preserve solution size. However, for parameterized complexity, one must also control how parameters behave under such reductions. For this reason,~\cite{fellows2015tractability} introduce a refined notion that requires, in addition to structural equivalence between instances, a linear correspondence between the associated parameters, ensuring that parameterized complexity results can be rigorously transferred. The formal definition is given below.

\begin{defi}[\cite{fellows2015tractability}]
    Let $\Pi$ and $\Pi'$ be analogous decision problems. The parameterized problems $\Pi(k_1, \ldots,k_t)$ and $\Pi'(k'_1, \ldots,k'_t)$ are said to be $p$-analogous if there exist linear-time reductions $f,g$ and a one-to-one correspondence $k_i \longleftrightarrow k'_i$ such that:

    \begin{enumerate}
        \item $\Pi(k_1, \ldots,k_t) \propto^{f} \Pi'(k'_1, \ldots,k'_t)$ and $\Pi'(k'_1, \ldots,k'_t) \propto^{g} \Pi(k_1, \ldots,k_t)$;

        \item every easily checkable certificate $\mathcal{C}$ for the yes-answer of the question ``$ I\in Y(\Pi(k_1, \ldots,k_t))$?'' implies an easily checkable certificate $\mathcal{C}'$ for the yes-answer of the question ``$ f(I)\in Y(\Pi'(k'_1, \ldots,k'_t))$ such that $k'_i=\varphi'_i(k_i)$  for some linear function $\varphi'_i(1\leq i \leq t)$.

        \item every easily checkable certificate $\mathcal{C}'$ for the yes-answer of the question ``$ I'\in Y(\Pi'(k'_1, \ldots,k'_t))$?'' implies an easily checkable certificate $\mathcal{C}$ for the yes-answer of the question ``$ g(I')\in Y(\Pi(k_1, \ldots,k_t))$ such that $k_i=\varphi_i(k'_i)$  for some linear function $\varphi_i(1\leq i \leq t)$.
    \end{enumerate}

\end{defi}

Before presenting Lemma~\ref{lemma:analogo-main}, it is worth noting that its proof has been fully rewritten to provide a clearer and more rigorous interpretation of the notion of analogicity introduced by Fellows et al. In particular, establishing that two problems are analogous cannot rely on a one-directional transformation alone; it requires demonstrating a complete correspondence between feasible solutions of both formulations. This necessarily involves proving both directions of the argument, showing that every feasible solution of the original instance maps to a feasible solution of the transformed instance, and vice versa.

The “forward-and-backward’’ structure of the proof is therefore not merely stylistic, but intrinsic to the formal definition of analogous problems, as it ensures the preservation of feasibility, solution size, and structural behavior under linear-time reductions. In the rewritten version that follows, this principle is made explicit, yielding a more transparent and theoretically faithful presentation of the result.

\begin{lemma}[\cite{gama2019aspects}]
Let $\mathcal{C}$ be a class of graphs closed under the operator $\psi$, that is, 
for every $G \in \mathcal{C}$ we have $\psi(G) \in \mathcal{C}$, where $\psi$ is the 
construction that, from an instance $(G,L)$ of \textit{List Coloring}, produces a 
graph $\psi(G)$ by adding, for each vertex $v$ and each color 
$i \notin L(v)$, a pendant vertex $w_i$ adjacent to $v$ with 
$\gamma(w_i) = \mu(w_i) = i$, and setting $\gamma(v)=1$, $\mu(v)=c$ for all 
$v \in V(G)$, where $c$ is the largest color appearing in the lists.

Then, when restricted to the class $\mathcal{C}$, the problems 
\textit{List Coloring} and $(\gamma,\mu)$-\textit{coloring} are analogous in the 
sense of Fellows et al.
\end{lemma}\label{lemma:analogo-main}

\begin{proof}
Let $(G,L)$ be an instance of \textit{List Coloring} with $G \in \mathcal{C}$. 
By definition of $\psi$, we construct $\psi(G)$ as follows. Let $c$ be the highest 
color appearing in the lists of $L$. For each $v \in V(G)$, set $\gamma(v)=1$ and 
$\mu(v)=c$ in $\psi(G)$. For each color $i \in \{1,\dots,c\} \setminus L(v)$, add a 
pendant vertex $w_i$ adjacent to $v$ with $\gamma(w_i)=\mu(w_i)=i$. Since 
$\mathcal{C}$ is closed under $\psi$, we have $\psi(G) \in \mathcal{C}$. The 
construction is clearly linear in the size of $(G,L)$. We show that $(G,L)$ admits a proper list-coloring if and only if $\psi(G)$ admits 
a proper $(\gamma,\mu)$-coloring.

($\Rightarrow$) Suppose there exists a proper coloring $c:V(G)\to\mathbb{N}$ such 
that $c(v)\in L(v)$ for all $v\in V(G)$. Define $c':V(\psi(G))\to\mathbb{N}$ by 
$c'(v)=c(v)$ for every $v\in V(G)$, and $c'(w_i)=i$ for every pendant vertex 
$w_i$. Since $c(v)\in L(v)$, we have $c(v)\neq i$ whenever $i\notin L(v)$, and 
hence $c'(v)\neq c'(w_i)$. Moreover, since $c$ was a proper coloring, adjacent 
original vertices maintain distinct colors. Finally, each constraint 
$\gamma(x)\le c'(x)\le \mu(x)$ holds by construction. Thus $c'$ is a valid 
$(\gamma,\mu)$-coloring of $\psi(G)$.

($\Leftarrow$) Conversely, suppose that $\psi(G)$ has a proper $(\gamma,\mu)$-coloring 
$c'$. For each pendant vertex $w_i$, we have $\gamma(w_i)=\mu(w_i)=i$, hence 
$c'(w_i)=i$. If $c'(v)=i$ for any $i\notin L(v)$, then $v$ and $w_i$ would receive 
the same color, contradicting properness. Therefore $c'(v)\in L(v)$ for all 
$v\in V(G)$, and adjacency constraints are preserved among original vertices. 
Hence the restriction of $c'$ to $V(G)$ is a valid list-coloring of $(G,L)$.

We have therefore shown a linear-time reduction from \textit{List Coloring} to 
$(\gamma,\mu)$-coloring that preserves solutions, certificates, and their size 
up to linear functions. The reverse reduction follows by observing that any 
instance of $(\gamma,\mu)$-coloring can be viewed as a list-coloring instance 
with interval lists $L(v)=\{\gamma(v),\gamma(v)+1,\dots,\mu(v)\}$. Thus, on any class of graphs closed under $\psi$, the problems 
\textit{List Coloring} and $(\gamma,\mu)$-coloring satisfy the analogicity 
conditions of Fellows et al.
\end{proof}

The theorem above shows that, on any graph class closed under the operator 
$\psi$, \textit{List Coloring} and $(\gamma,\mu)$-coloring satisfy the strong 
equivalence required by Fellows et al. Hence, in such classes, both problems 
are interchangeable from a complexity-theoretic perspective. This immediately 
implies that any algorithmic or parameterized result proved for one of them 
must also hold for the other. The corollary below formalizes this fact.

\begin{corollary}
Let $\mathcal{C}$ be any graph class closed under the operator $\psi$. 
If a computational or parameterized complexity result holds for 
\textit{List Coloring} on $\mathcal{C}$, then the same result holds for 
$(\gamma,\mu)$-\textit{coloring} on $\mathcal{C}$, and vice versa. In particular, if \textit{List Coloring} is polynomial-time solvable (resp.\ 
NP-complete, FPT under some parameter, W[1]-hard, etc.) on $\mathcal{C}$, then 
$(\gamma,\mu)$-\textit{coloring} inherits exactly the same complexity behavior 
on $\mathcal{C}$.
\end{corollary}

\begin{proof}
Since the theorem establishes that \textit{List Coloring} and 
$(\gamma,\mu)$-coloring are analogous on $\mathcal{C}$, there exist linear-time 
reductions $f$ and $g$ between the two problems satisfying the analogicity 
conditions of Fellows et al. These reductions preserve yes-instances, feasible 
solutions, and the size of certificates up to linear functions.

Therefore, any algorithmic upper bound for one problem translates directly 
into an algorithmic upper bound for the other via composition with $f$ or $g$, 
and any hardness or intractability result transfers in the opposite direction. 
The same holds for parameterized reductions because the reductions do not 
increase parameter values beyond linear transformations. Hence all complexity and parameterized properties shared through such 
reductions are preserved bidirectionally between the two problems on 
$\mathcal{C}$.
\end{proof}

To extend the previous correspondence to the parameterized framework, we next  
establish that the relevant coloring problems preserve their structural  
relationship under standard width parameters. The following lemma formalizes  
this.

\begin{lemma}\label{lemma:p-Analogous-final}
Let $\mathcal{C}$ be any hereditary graph class that is closed under the addition  
of pendant vertices. Then the problems \textsc{List Coloring},  
$(\gamma,\mu)$\textsc{-Coloring}, and \textsc{Precoloring Extension}, when  
restricted to $\mathcal{C}$, are $p$-analogous under the parameters  
\emph{treewidth} and \emph{feedback vertex set}.
\end{lemma}

\begin{proof}
By Lemma~\ref{lemma:analogo-main}, \textsc{List Coloring} and  
$(\gamma,\mu)$\textsc{-Coloring} are analogous when instances are transformed  
using the operator $\psi$. Since $\mathcal{C}$ is hereditary and closed under the  
addition of pendant vertices, we have $\psi(G) \in \mathcal{C}$ for every  
$G \in \mathcal{C}$.

To verify $p$-analogicity, we show that both treewidth and the size of a minimum  
feedback vertex set are preserved under $\psi$. The operator $\psi$ augments $G$  
only by attaching pendant vertices, which neither creates cycles nor increases  
the width of any tree decomposition. Therefore:
\[
\mathrm{tw}(G) = \mathrm{tw}(\psi(G)) \quad \text{and} \quad 
\mathrm{fvs}(G) = \mathrm{fvs}(\psi(G)).
\]

The forward and backward reductions map feasible solutions to feasible solutions  
of identical size, satisfying the requirements in the definition of analogous and  
$p$-analogous problems. Since the parameters are preserved exactly, the three  
problems are $p$-analogous on $\mathcal{C}$.
\end{proof}

\begin{theorem}\label{Theo:GammaMu-final}
The problems $(\gamma,\mu)$\textsc{-Coloring} and \textsc{Precoloring Extension}  
parameterized by the size of a minimum feedback vertex set are $W[1]$-hard even  
when restricted to bipartite graphs.
\end{theorem}

\begin{proof}
By \cite{fellows2011complexity}, \textsc{List Coloring} is $W[1]$-hard on  
bipartite graphs when parameterized by the feedback vertex set.  
Let $\mathcal{C}$ be the class of bipartite graphs. Since the operator $\psi$  
adds only pendant vertices to $G$, it preserves bipartiteness; hence  
$\mathcal{C}$ is closed under $\psi$.

By Lemma~\ref{lemma:p-Analogous-final}, on any class closed under $\psi$ the  
problems \textsc{List Coloring}, $(\gamma,\mu)$\textsc{-Coloring}, and  
\textsc{Precoloring Extension} are $p$-analogous under the feedback vertex set  
parameter. Therefore, $W[1]$-hardness transfers from \textsc{List Coloring} to  
both $(\gamma,\mu)$\textsc{-Coloring} and \textsc{Precoloring Extension} on the  
class of bipartite graphs.
\end{proof}

Since $\mu$-coloring is a syntactic restriction of  
$(\gamma,\mu)$\textsc{-coloring} (obtained by fixing $\gamma(v)=1$ for every  
vertex), its $W[1]$-hardness under the feedback vertex set parameter also holds.

We now turn to $\mu$-coloring. In contrast to the equivalence obtained for 
\textsc{List Coloring} and $(\gamma,\mu)$\textsc{-coloring} on $\psi$-closed 
graph classes, $\mu$-coloring does not admit the same kind of structural 
flexibility. In particular, there exist graph classes on which \textsc{List 
Coloring} and $\mu$-coloring have fundamentally different computational 
behaviour, and this alone is enough to rule out any analogicity between them 
in the sense of Fellows et al.

Throughout this subsection, let $\mathcal{C}$ be a fixed graph class and consider 
the restrictions of \textsc{List Coloring} and $\mu$-coloring to instances whose 
underlying graph lies in $\mathcal{C}$. We assume:

\begin{itemize}
    \item (H1a)\textsc{List Coloring} is NP-complete on $\mathcal{C}$;
    \item (H1b) $\mu$-coloring is solvable in polynomial time on $\mathcal{C}$;
    \item (H2) $P \neq NP$.
\end{itemize}

Typical examples of such a class $\mathcal{C}$ include clique-trees of 
height~2, where \textsc{List Coloring} is NP-complete while $\mu$-coloring is 
polynomial-time solvable (\cite{bonomo2012coloring}), and, as we shall see, 
the class of cographs under the algorithmic results proved later in this work.

We say that the restricted problems \textsc{List Coloring}$\!\upharpoonright\!\mathcal{C}$ 
and $\mu$-coloring$\!\upharpoonright\!\mathcal{C}$ are analogous on $\mathcal{C}$ if there 
exist linear-time reductions between them satisfying the analogicity conditions 
of Fellows \textit{et al.}\ when both the input and the output graphs are required to 
belong to~$\mathcal{C}$.

\begin{theorem}\label{thm:no-analogicity-mu-list}
Under assumptions \textnormal{(H1a)--(H2)}, the restricted problems 
\textnormal{\textsc{List Coloring}}$\!\upharpoonright\!\mathcal{C}$ and 
\textnormal{$\mu$-coloring}$\!\upharpoonright\!\mathcal{C}$ cannot be analogous on 
$\mathcal{C}$ in the sense of Fellows et al. In particular, there is no pair of linear-time reductions 
$f : \textsc{List Coloring}\!\upharpoonright\!\mathcal{C} \to 
\mu\text{-coloring}\!\upharpoonright\!\mathcal{C}$ 
and 
$g : \mu\text{-coloring}\!\upharpoonright\!\mathcal{C} \to 
\textsc{List Coloring}\!\upharpoonright\!\mathcal{C}$
satisfying the analogicity conditions. Hence, on $\mathcal{C}$, 
\textnormal{\textsc{List Coloring}} and \textnormal{$\mu$-coloring} are not analogous problems.
\end{theorem}
\begin{proof}

By (H1a) and (H1b), \textsc{List Coloring}$\!\upharpoonright\!\mathcal{C}$ is 
NP-complete, whereas $\mu$-coloring$\!\upharpoonright\!\mathcal{C}$ is solvable 
in polynomial time. Assume, for contradiction, that the two restricted 
problems are analogous on $\mathcal{C}$ in the sense of Fellows et al., that is, 
there exist linear-time reductions $f$ and $g$ between them that preserve 
yes-instances, feasible solutions, and certificate sizes up to linear functions, 
and that both $f$ and $g$ map instances whose underlying graph lies in 
$\mathcal{C}$ to instances whose underlying graph also lies in $\mathcal{C}$.

Consider any instance $(G,L)$ of \textsc{List Coloring} with $G \in \mathcal{C}$. 
Applying $f$ yields, in linear time, an instance $(G',\mu') = f(G,L)$ of 
$\mu$-coloring whose underlying graph still belongs to $\mathcal{C}$. By (H1b), 
$\mu$-coloring$\!\upharpoonright\!\mathcal{C}$ can be solved in polynomial time, 
so we can decide in polynomial time whether $(G',\mu')$ is a yes-instance of 
$\mu$-coloring. By correctness of the reduction $f$, the instance $(G,L)$ is a 
yes-instance of \textsc{List Coloring}$\!\upharpoonright\!\mathcal{C}$ if and only 
if $(G',\mu')$ is a yes-instance of $\mu$-coloring$\!\upharpoonright\!\mathcal{C}$.

Therefore, composing $f$ with the polynomial-time algorithm for 
$\mu$-coloring$\!\upharpoonright\!\mathcal{C}$ yields a polynomial-time algorithm 
for \textsc{List Coloring}$\!\upharpoonright\!\mathcal{C}$. This contradicts (H1a), 
since \textsc{List Coloring} is NP-complete on $\mathcal{C}$ and, by (H2), 
$P \neq NP$. Hence, no such reductions $f$ and $g$ can exist, and the two 
problems are not analogous on $\mathcal{C}$.
\end{proof}

\subsection{Analogous problems in $k$-choosability}

It is unlikely that any analogicity result, in the sense of \cite{fellows2015tractability}, can hold between 
$k$-choosability and DP-$k$-choosability. The obstruction is structural: $k$-choosability is defined by the logical pattern
\[
\forall L\; (|L(v)|=k)\; \exists \varphi,
\]
whereas DP-$k$-choosability adds an additional universal quantification over all compatible DP-covers. Thus, the expressive 
power of DP-$k$-choosability is strictly stronger, and it is highly implausible that bijective linear-time reductions 
preserving solution structure could exist. This difficulty disappears when one replaces $k$-choosability by 
$k$-colorability: under classical covers, DP-$k$-colorability collapses to ordinary $k$-coloring.

The problems of $k$-choosability and $(k:1)$-choosability arise as specific instances within the general framework of 
list-based coloring constraints. Although defined differently, both enforce the same requirement: for every list assignment 
of size $k$, one must be able to pick exactly one color from each list so as to obtain a proper coloring. The formal 
definitions follow.

\begin{defi}
Fix $k \ge 1$. For a graph $G$:

\begin{itemize}
    \item $G$ is \emph{$k$-choosable} if, for every list assignment 
    $L$ with $|L(v)| = k$ for all $v$, there exists a proper coloring 
    $\varphi$ with $\varphi(v)\in L(v)$.

    \item $G$ is \emph{$(k:1)$-choosable} if, for every list assignment 
    $L$ with $|L(v)| = k$, there exists a mapping $\Phi$ with 
    $|\Phi(v)| = 1$, $\Phi(v)\subseteq L(v)$, and 
    $\Phi(u)\cap\Phi(v)=\emptyset$ whenever $uv\in E(G)$.
\end{itemize}

Since a singleton choice $\Phi(v)=\{\varphi(v)\}$ is equivalent to selecting a single color, both formalisms encode the same combinatorial constraint.
\end{defi}

\begin{theorem}
For fixed $k$, $k$-choosability and $(k:1)$-choosability are analogous problems in the sense of 
\cite{fellows2015tractability}: there exist linear-time reductions in both directions, and every certificate 
for one formulation corresponds exactly to a certificate for the other with identical size.
\end{theorem}

\begin{proof}
Fix $k\ge1$. We first prove that $G$ is $k$-choosable if and only if it is $(k:1)$-choosable.

Let $L$ be any list assignment with $|L(v)|=k$. If $G$ is $k$-choosable, there exists a proper list-coloring 
$\varphi$. Define $\Phi(v)=\{\varphi(v)\}$. Then $|\Phi(v)|=1$, $\Phi(v)\subseteq L(v)$, and adjacency implies 
disjointness. Thus $G$ is $(k:1)$-choosable.

Conversely, suppose $G$ is $(k:1)$-choosable. Given $L$, let $\Phi$ be a valid witness. For each $v$, let 
$\varphi(v)$ be the unique element of $\Phi(v)$. Properness follows from disjointness. Hence $G$ is k-choosable.

Thus $G$ is a Yes-instance of one problem iff it is a Yes-instance of the other.

\medskip\noindent\textbf{Reductions.}
Instances of both problems are simply graphs. Define
\[
f(G)=G, \qquad g(G)=G.
\]
These are linear-time and preserve Yes/No answers by the equivalence above.

\medskip\noindent\textbf{Certificates.}
In both formulations (under a fixed list assignment $L$), a certificate is simply a choice of one color per 
vertex. Under a standard encoding, both certificates have exactly one color per vertex, hence the same size. 
The transformations
\[
\varphi \mapsto \Phi(v)=\{\varphi(v)\}, \qquad 
\Phi \mapsto \varphi(v) \text{ (the unique element of }\Phi(v))
\]
are bijective and computable in linear time.

\medskip\noindent\textbf{Parameterized version.}
Let $\Pi(G,\kappa(G))$ and $\Pi'(G,\kappa(G))$ denote the parameterized versions of $k$-choosability and 
$(k:1)$-choosability with respect to any graph parameter $\kappa$ (treewidth, maximum degree, feedback vertex 
number, etc.). Since the reductions $f$ and $g$ are the identity on $G$, they preserve $\kappa(G)$ exactly, 
and the correspondence between certificates is size-preserving. Therefore, by the definition of $p$-analogous 
problems in \cite{fellows2015tractability}, the parameterized problems 
$(k\text{-choosability},\kappa)$ and $((k:1)\text{-choosability},\kappa)$ are $p$-analogous.

\end{proof}

\begin{corollary}
For any fixed $k$ and any graph class $\mathcal{C}$, deciding $k$-choosability and deciding $(k:1)$-choosability 
on graphs of $\mathcal{C}$ have identical complexity classifications. In particular, $k$-choosability is 
polynomial-time solvable (respectively, NP-complete, $\Pi_2^P$-complete) on $\mathcal{C}$ if and only if 
$(k:1)$-choosability is polynomial-time solvable (respectively, NP-complete, $\Pi_2^P$-complete) on $\mathcal{C}$. 
The same equivalence holds for all parameterized versions with respect to parameters depending only on $G$.
\end{corollary}

\section{Restricted \texorpdfstring{$k$}{k}-choosability in Graphs}\label{sec:restricted-k-choosable}

In this work, we build upon the notion of $k$-$(\gamma,\mu)$-choosability, originally introduced in \cite{gama2018choosability}, in which the classical $k$-choosability framework is restricted by requiring that all admissible lists belong to a structured family of subsets. This restriction models scenarios where only highly organized color combinations are feasible, reflecting practical contexts of resource allocation subject to external constraints.

As a necessary preliminary step, we first formalize the $k$-$(\gamma,\mu)$-coloring problem, in which each vertex $v$ receives an admissible list consisting of a consecutive interval of integers of fixed size $k$; that is, an interval $[\gamma(v), \mu(v)]$ satisfying $\mu(v) - \gamma(v) + 1 = k$. The analysis of this restricted coloring variant is essential for two main reasons. First, it allows us to isolate the structural impact of imposing consecutive-interval lists, clearly separating it from the additional phenomena introduced by the universal quantification inherent to choosability. Second, the characterization of $k$-$(\gamma,\mu)$-coloring provides formal support for subsequent results, since several algorithmic properties of this restricted version do not extend to the general $(\gamma,\mu)$-coloring problem.

Only after establishing this conceptual foundation do we proceed to the definition of $k$-$(\gamma,\mu)$-choosability, which requires the graph to be colorable under \emph{every} possible assignment of consecutive intervals of size $k$. This introduces a substantially higher level of complexity, as the problem now incorporates universal quantification over all admissible list assignments, placing it naturally within higher levels of the complexity hierarchy. Therefore, the distinction between $k$-$(\gamma,\mu)$-coloring and $k$-$(\gamma,\mu)$-choosability is not merely methodological but necessary to support the computational complexity analysis developed in this article and to provide an appropriate initial characterization of this interval-constrained choosability model.

\begin{defi}
Let $G = (V, E)$ be a simple graph. An \emph{$k$-$(\gamma, \mu)$-coloring} of $G$ is a proper vertex coloring obtained from a list assignment $L$, where for each vertex $v \in V$:
\[
L(v) = \{\gamma(v), \gamma(v)+1, \dots, \mu(v)\}
\]
such that:
\begin{enumerate}
    \item The list is an integer interval of fixed size $k$: $\mu(v) - \gamma(v) + 1 = k$;
    \item The coloring $c: V \to \mathbb{N}$ satisfies $c(v) \in L(v)$ for all $v \in V$;
    \item Adjacent vertices receive different colors: $c(u) \neq c(v)$ for all $(u, v) \in E$.
\end{enumerate}
We say that $G$ is \emph{$k$-$(\gamma, \mu)$-colorable} if such a coloring exists under every assignment of integer intervals of size $k$ to the vertices.
\end{defi}

Thus, \(c_{(\gamma, \mu)}(G)\) represents a valid coloring of the graph \(G\) that respects the constraints imposed by the integer intervals \([\gamma(v), \mu(v)]\) assigned to each vertex.

The $k$-$(\gamma,\mu)$-coloring problem can be regarded as a syntactically restricted version of the $(\gamma,\mu)$-coloring problem, in which every admissible list assigned to a vertex is required to be an interval of fixed length $k$. Formally, for each vertex $v \in V(G)$, the associated list satisfies $\mu(v)-\gamma(v)+1 = k$.

Every instance of $k$-$(\gamma,\mu)$-coloring is therefore trivially an instance of $(\gamma,\mu)$-coloring. The converse, however, does not hold in general, since $(\gamma,\mu)$-coloring allows intervals of arbitrary length. As a consequence, there is no general linear-time reduction from $(\gamma,\mu)$-coloring to $k$-$(\gamma,\mu)$-coloring that preserves instances and solutions.

Moreover, the restriction to fixed-length intervals leads to fundamentally different algorithmic behavior. In particular, several graph classes for which $(\gamma,\mu)$-coloring remains NP-complete admit polynomial-time algorithms for $k$-$(\gamma,\mu)$-coloring when $k$ is fixed. This separation in complexity precludes the two problems from being analogous in the sense of Fellows \textit{et al}., despite their close syntactic relationship.

To analyze the computational requirements of the interval-based model, we first formalize the corresponding decision problem:

\noindent
\underline{\textsc{$k$-$(\gamma,\mu)$-coloring}}

\noindent
\textbf{Instance:} A graph $G$ and a list assignment $L$ for $G$ such that
$L(v) = [\gamma(v), \mu(v)]$ and $\mu(v) - \gamma(v) + 1 = k$ for all $v \in V(G)$.

\noindent
\textbf{Question:} Is there a proper coloring of $G$ that respects $L$?

\begin{proposition}
For any fixed integer $k \ge 3$, the \textsc{$k$-$(\gamma,\mu)$-coloring} problem
is NP-complete on general graphs.
\end{proposition}

\begin{proof}
Membership in NP is immediate, since a coloring of $G$ can be verified in polynomial
time by checking whether adjacent vertices receive distinct colors and whether each
vertex $v$ is assigned a color in its interval $L(v)$.

To prove NP-hardness, we reduce from the classical \textsc{$k$-coloring} problem,
which is NP-complete for $k \ge 3$.
Given an arbitrary graph $G$, we construct an instance of
\textsc{$k$-$(\gamma,\mu)$-coloring} by assigning to every vertex $v \in V(G)$
the same interval list $L(v) = [1,k]$.
This construction is clearly computable in polynomial time. Observe that $G$ admits a proper $k$-coloring if and only if it admits a
\textsc{$k$-$(\gamma,\mu)$-Coloring} respecting the lists $L(v)=[1,k]$ for all $v$.
Therefore, the reduction is correct. Since \textsc{$k$-coloring} is NP-complete, it follows that
\textsc{$k$-$(\gamma,\mu)$-coloring} is NP-complete as well.
\end{proof}

A fundamental property in graph coloring is that a graph admitting a proper $k$-coloring also admits a proper $(\gamma, \mu)$-coloring where each list has size $k$, i.e., a $k$-$(\gamma, \mu)$-coloring. In this context, each vertex is assigned a list of $k$ consecutive colors determined by the parameters $(\gamma(v), \mu(v))$, and the goal is to select one color from each list to construct a proper coloring of the graph.

The key observation is that the existence of a proper $k$-coloring guarantees that, regardless of how the $(\gamma, \mu)$-lists of size $k$ are assigned, it is always possible to choose colors from those intervals such that adjacent vertices receive different colors. This relationship is formally established in Theorem~\ref{teo:general}.

\begin{theorem}\label{teo:general}
Let $G$ be a $k$-colorable graph. Then $G$ admits a proper coloring from $(\gamma, \mu)$-lists of size $k$ assigned to its vertices, where each list is an interval of $k$ consecutive integers. Consequently, every $k$-colorable graph also admits a $k$-$(\gamma, \mu)$-coloring, in which the parameter $k$ corresponds to the fixed size of the color interval assigned to each vertex.
\end{theorem}

\begin{proof}
    When dividing any integer \( m \) by \( k \), we obtain a remainder \( r \) such that \( r \in \{0, 1, \ldots, k-1\} \). This division naturally partitions the integers into \( k \) distinct sets. Define the sets as follows:
\[
\begin{aligned}
A_0 &= \{n \in \mathbb{Z} \mid n \equiv 0 \pmod{k}\}, \\
A_1 &= \{n \in \mathbb{Z} \mid n \equiv 1 \pmod{k}\}, \\
&\quad \vdots \\
A_{k-1} &= \{n \in \mathbb{Z} \mid n \equiv k-1 \pmod{k}\}.
\end{aligned}
\]
These sets are pairwise disjoint, that is, \( A_i \cap A_j = \emptyset \) for all \( i \neq j \).

An important property of these sets is that in any sequence of \( k \) consecutive integers \( x, x+1, \ldots, x+k-1 \), exactly one integer belongs to each set \( A_i \). In other words, each remainder \( i \) (for \( i = 0, 1, \ldots, k-1 \)) appears exactly once among the remainders of these \( k \) integers modulo \( k \).

Now, consider a given assignment of lists of size \( k \) of \((\gamma, \mu)\)-types to the vertices of \( V \). Specifically, for each vertex \( v_i \in V \), define
\[
L(v_i) = \{\gamma_i, \gamma_i+1, \ldots, \gamma_i+k-1\}.
\]

Let \( v_i \) and \( v_j \) be adjacent vertices in \( G \), and suppose that \( c_i = c(v_i) \) and \( c_j = c(v_j) \) are their respective colors in a proper \( k \)-coloring of \( G \), so that \( c_i \neq c_j \).

Since \( c_i \) and \( c_j \) correspond to different sets \( A_{c_i} \) and \( A_{c_j} \), and these sets are pairwise disjoint, we can construct a \((\gamma, \mu)\)-list coloring as follows:  
assign to \( v_i \) an element of \( L(v_i) \) that belongs to \( A_{c_i} \), and assign to \( v_j \) an element of \( L(v_j) \) that belongs to \( A_{c_j} \).

This choice ensures that adjacent vertices receive distinct colors, because the elements selected belong to disjoint sets.  
Thus, \( G \) admits a proper \( k \)-\((\gamma, \mu)\)-coloring.
\end{proof}

In Theorem~\ref{teo:general}, we established that every $k$-colorable graph is also $k$-$(\gamma, \mu)$-colorable. Theorem~\ref{teo:general_complexity} further shows that, given a known proper $k$-coloring of the graph, a corresponding $k$-$(\gamma, \mu)$-coloring can be constructed efficiently, in $O(k|V|)$ time, by selecting for each vertex a color from its assigned interval based on its original color class.

\begin{theorem}
The coloring construction described in Theorem~\ref{teo:general} can be performed in \( O(k|V|) \) time, where \( |V| \) is the number of vertices of the graph.
\end{theorem}\label{teo:general_complexity}

\begin{proof}
Given a \( k \)-coloring \( c: V \to \{0, 1, \dots, k-1\} \) of the graph \( G \), and a list \( L(v) \) of exactly \( k \) colors for each vertex \( v \in V \), the coloring method in Theorem~1 assigns a final color to each vertex based on the modulo operation.

For each vertex \( v \), it suffices to inspect the \( k \) elements in \( L(v) \) and select the one corresponding to the residue class determined by \( c(v) \mod k \). Since computing the remainder of an integer modulo \( k \) can be done in constant time, and checking up to \( k \) elements requires \( O(k) \) time per vertex, the total time for coloring all vertices is \( O(k|V|) \).

Thus, the construction described in Theorem~\ref{teo:general} can be completed in \( O(k|V|) \) time.
\end{proof}

To illustrate the construction described in Theorem~\ref{teo:general}, consider the following example. We take a simple graph that is $k$-colorable and demonstrate how to derive a corresponding $k$-$(\gamma, \mu)$-coloring by assigning appropriate intervals to each vertex and selecting colors based on modular residue classes.

\begin{example}
Let $G$ be a cycle graph on four vertices, $C_4 = (v_1, v_2, v_3, v_4)$, with edges $(v_1, v_2)$, $(v_2, v_3)$, $(v_3, v_4)$ and $(v_4, v_1)$. This graph is $2$-colorable, and one possible proper coloring is:
\[
c(v_1) = 1,\quad c(v_2) = 2,\quad c(v_3) = 1,\quad c(v_4) = 2.
\]

To construct a $2$-$(\gamma, \mu)$-coloring, assign to each vertex a list of two consecutive integers as follows:
\[
\begin{aligned}
L(v_1) &= \{10, 11\},\\
L(v_2) &= \{20, 21\},\\
L(v_3) &= \{30, 31\},\\
L(v_4) &= \{40, 41\}.
\end{aligned}
\]

We now select a color from each list according to the residue class modulo $2$ (based on the original coloring):

\begin{itemize}
    \item $c(v_1) = 1$ implies selection of an odd number from $L(v_1)$: choose $11$;
    \item $c(v_2) = 2$ implies selection of an even number from $L(v_2)$: choose $20$;
    \item $c(v_3) = 1$ implies selection of an odd number from $L(v_3)$: choose $31$;
    \item $c(v_4) = 2$ implies selection of an even number from $L(v_4)$: choose $40$.
\end{itemize}

Since the selected values are in distinct congruence classes and respect the adjacency constraints, this forms a valid $2$-$(\gamma, \mu)$-coloring of $G$.
\end{example}

Since a $(\gamma, \mu)$-coloring with lists of size $k$ (that is, a $k$-$(\gamma, \mu)$-coloring) can be efficiently constructed from a proper $k$-coloring, we can conclude that, for graph classes where the $k$-coloring problem is solvable in polynomial time, the corresponding $k$-$(\gamma, \mu)$-coloring problem, in which each vertex receives a list of $k$ consecutive colors, is also solvable in polynomial time.

\begin{corollary}
Let \(\mathcal{G}\) be a class of graphs for which the \(k\)-coloring problem can be solved in polynomial time. Then, for any graph \(G \in \mathcal{G}\), the \(k\)-\((\gamma, \mu)\)-coloring problem can also be solved in polynomial time.
\end{corollary}\label{corol:complexity-kgammamu}

\begin{proof}
By Theorem~\ref{teo:general}, if a graph \(G\) admits a proper \(k\)-coloring, then it also admits a proper \(k\)-\((\gamma, \mu)\)-coloring. Moreover, by Theorem~2, given a \(k\)-coloring of \(G\), a corresponding \(k\)-\((\gamma, \mu)\)-coloring can be constructed in \(O(k|V|)\) time. Since, by hypothesis, the \(k\)-coloring of graphs in \(\mathcal{G}\) can be found in polynomial time, and the transformation to a \(k\)-\((\gamma, \mu)\)-coloring is also polynomial, it follows that the overall process for obtaining a \(k\)-\((\gamma, \mu)\)-coloring is polynomial for graphs in \(\mathcal{G}\).
\end{proof}

In \textbf{Table~\ref{tab:coloring_complexity}}, we summarize the computational complexity of different coloring problems for specific graph classes. Although the classical \(k\)-coloring problem can be solved in polynomial time for bipartite, split, and interval graphs, the \((\gamma, \mu)\)-coloring problem remains NP-complete even for these restricted classes \citep{bonomo2009exploring}. However, when the number of colors \(k\) is fixed, the \(k\)-\((\gamma, \mu)\)-coloring problem can still be solved in polynomial time, following from Corollary~\ref{corol:complexity-kgammamu}.

\begin{table}[h!]
\small
\caption{Complexity of \((\gamma, \mu)\)-coloring and \(k\)-\((\gamma, \mu)\)-coloring in different graph classes. The entries in bold under the \(k\)-\((\gamma, \mu)\)-coloring column indicate cases where the problem is solvable in polynomial time (\textbf{P}). The gray background highlights results directly derived from this Corollary~\ref{corol:complexity-kgammamu}.}

\centering
\begin{tabular}{|c|c|>{\columncolor{gray!15}}c|}
\hline\hline
\textbf{Graph Class} & \textbf{\((\gamma, \mu)\)-coloring} & \textbf{\(k\)-\((\gamma, \mu)\)-coloring} \\
\hline
Bipartite          & NP-C & \textbf{P} \\ \hline
Split              & NP-C & \textbf{P} \\ \hline
Interval           & NP-C & \textbf{P} \\ \hline
Unit Interval      & NP-C & \textbf{P} \\ \hline
Line of $K_{n,n}$  & NP-C & \textbf{P} \\ \hline
Line of $K_{n}$    & NP-C & \textbf{P} \\ \hline
\hline
\end{tabular}
\label{tab:coloring_complexity}
\end{table}

\subsection{The \texorpdfstring{$k$}{k}-$(\gamma,\mu)$-choosability in Graphs}

Building upon the notion of $k$-$(\gamma,\mu)$-coloring, in which each vertex
receives an interval list of exactly $k$ consecutive integers, we now extend our
analysis to a stronger requirement: instead of considering a single fixed list
assignment, we require the graph to be colorable under \emph{every} possible
assignment of such interval lists. This leads to the notion of
\emph{$k$-$(\gamma,\mu)$-choosability}.

This leads us to introduce the notion of \emph{$k$-$(\gamma, \mu)$-choosability}, in which a graph must admit a proper coloring for \textit{every} possible assignment of $(\gamma, \mu)$-intervals of size $k$ to its vertices. In this context, the parameter $k$ refers to the fixed length of the list (interval) associated with each vertex. This property can be interpreted as requiring that the graph admits an $r$-$(\gamma, \mu)$-coloring with $r = k$ for all such list assignments.

\begin{defi} A graph $G$ is said to be \emph{$k$-$(\gamma,\mu)$-choosable} if, for every assignment of intervals $[\gamma(v), \mu(v)]$ of size $k$ to each vertex $v \in V$, there exists a proper coloring that selects one color from each list such that adjacent vertices receive distinct colors. In other words, $G$ admits an $r$-$(\gamma, \mu)$-coloring with $r = k$ for every such list assignment.
\end{defi}

The smallest integer $k$ for which a graph $G$ is $k$-$(\gamma,\mu)$-choosable is called the \textit{$(\gamma,\mu)$-choosability number} of $G$, denoted by $\chi_{\ell,(\gamma,\mu)}(G)$.

As an example, consider the graph $K_3$, which is 3-colorable, 3-choosable, and also 3-$(\gamma,\mu)$-choosable. This means that $K_3$ can be properly colored for any assignment of color lists of type $(\gamma,\mu)$ with size three associated to its vertices.

\begin{example}
Consider the complete graph \(K_3\) with vertices \(v_1\), \(v_2\), and \(v_3\). We aim to show that \(K_3\) is \(3\)-\((\gamma, \mu)\)-choosable.

To prove this, we must show that for every assignment of color lists of the form \([\gamma(v), \mu(v)]\) of size three (i.e., lists of three consecutive integers), there exists a proper \((\gamma, \mu)\)-coloring using one color from each list.

As an example, consider the following assignment of lists:
\[
L(v_1) = \{4, 5, 6\}, \quad L(v_2) = \{1, 2, 3\}, \quad L(v_3) = \{7, 8, 9\}.
\]

Each list corresponds to an interval \([\gamma(v), \mu(v)]\) of size three, satisfying the \((\gamma, \mu)\) structure. A proper coloring can be achieved by assigning:
\[
v_1 \rightarrow 4, \quad v_2 \rightarrow 2, \quad v_3 \rightarrow 8.
\]

Since \(K_3\) is 3-colorable, and this holds for any such assignment of \((\gamma, \mu)\)-structured lists of size 3, we conclude that \(K_3\) is \(3\)-\((\gamma, \mu)\)-choosable.
\end{example}

\noindent
\underline{\textsc{$k$-$(\gamma,\mu)$-choosability}}

\noindent
\textbf{Instance:} A graph $G = (V,E)$.

\noindent
\textbf{Question:} Is $G$ $k$-$(\gamma,\mu)$-choosable?

Note that, in contrast to coloring problems, the list assignment does not form part of the input. By definition, a graph is $k$-$(\gamma,\mu)$-choosable if it admits a proper coloring for \emph{every} possible assignment of consecutive integer intervals of size $k$ to its vertices. Thus, the decision problem quantifies universally over all admissible list assignments, rather than fixing a single instance.

We emphasize that, unlike in list coloring, the list assignment does not appear
in the input of the problem. By definition, a graph is $k$-$(\gamma,\mu)$-choosable
if it admits a proper coloring for \emph{every} possible assignment of consecutive
integer intervals of size $k$ to its vertices. Consequently, the corresponding
decision problem quantifies universally over all admissible list assignments,
rather than fixing a single instance.

This universal quantification has direct implications for computational complexity:
determining whether $G$ is $k$-$(\gamma,\mu)$-choosable can be expressed as
\[
\forall L \;\exists c \; R(G,k,L,c),
\]
where $R(G,k,L,c)$ verifies in polynomial time that $c$ is a proper coloring
consistent with the interval lists $L$. The alternation of quantifiers 
$\forall L$ followed by $\exists c$ gives the problem the characteristic
logical structure of $\Pi_2^P$ problems in the polynomial hierarchy.
Thus, the definition itself inherently places the decision problem at the
second level of the hierarchy, independently of any particular algorithmic
approach used to test choosability.

In Table~\ref{tab:choosability-structure-complexity}, we compare the computational complexity of different choosability variants according to their list-structure restrictions.

\begin{table*}[ht]
\centering
\caption{
Comparison of list-structure constraints and computational complexity for several
choosability generalizations. Unlike classical variants, which operate on fully
arbitrary list assignments, \emph{$k$-$(\gamma,\mu)$-choosability} restricts all
lists to consecutive intervals of fixed size $k$. Despite this structural
restriction, the associated decision problem retains the characteristic
$\Pi_2^P$-completeness arising from universal quantification over admissible list
assignments. Highlighted rows correspond to the interval-based choosability model
studied in this work.
}
\label{tab:choosability-structure-complexity}

\begin{NiceTabular}{
    p{4cm}
    p{6.5cm}
    >{\centering\arraybackslash}p{5cm}
}[hvlines]
\textbf{Generalization} & \textbf{List Structure} & \textbf{Computational Complexity} \\

(a:b)-choosability &
Arbitrary lists, $|L(v)| = a$ &
$\Pi_2^P$-complete \\

Defective choosability &
Arbitrary lists, $|L(v)| \ge k$ &
$\Pi_2^P$-complete \\

DP-choosability &
Arbitrary lists with arbitrary correspondences &
$\Pi_2^P$-complete \\

Online choosability &
Arbitrary lists revealed adversarially &
PSPACE (full classification open) \\

\rowcolor{gray!15}
$k$-$(\gamma,\mu)$-choosability &
Interval lists of fixed size $k$ &
$\Pi_2^P$-complete \\

\end{NiceTabular}
\end{table*}

\vspace{0.5cm}
\noindent\textbf{Application of $k$-$(\gamma,\mu)$-choosable.}
A practical setting where $k$-$(\gamma,\mu)$-choosability naturally arises is the
assignment of operating frequencies in wireless communication networks. Each
transmitter is represented by a vertex of a graph, and edges correspond to
pairs of transmitters whose signal ranges overlap, meaning they cannot operate
on the same frequency without causing interference. Due to hardware
constraints, regulatory limitations, or environmental noise, each transmitter
can only operate within a \emph{consecutive} interval of admissible frequencies,
represented by $L(v) = \{\gamma(v),\gamma(v)+1,\dots,\mu(v)\}$. Selecting a
frequency for each transmitter thus amounts to choosing a single value inside
its interval list.

For illustration, consider three devices arranged linearly, modeled as the path
$A$--$B$--$C$, where $A$ and $B$ interfere, and $B$ and $C$ interfere, but $A$
and $C$ do not. Suppose each device is restricted to an interval list of size
$k=2$. One possible assignment is $L(A)=\{1,2\}$, $L(B)=\{2,3\}$, and
$L(C)=\{3,4\}$. A feasible choice of operating frequencies is then obtained by
assigning $1$ to $A$, $3$ to $B$, and $4$ to $C$, ensuring that adjacent
vertices receive different colors while respecting the interval constraints.

The relevance of $k$-$(\gamma,\mu)$-choosability in this context is that the
property guarantees the existence of a valid interference-free configuration
\emph{for every possible assignment} of interval lists of size $k$. Thus, if a
network is $k$-$(\gamma,\mu)$-choosable, it remains operable regardless of
local variations in available frequency bands, making the model suitable for
robust wireless systems subject to dynamic spectral conditions.

\subsubsection{Exact Decision Algorithms for \texorpdfstring{$k$}{k}-$(\gamma,\mu)$-choosability in Graphs}
\label{subsubsec:k-gamma-mu-choosable}

We present two exact algorithms to decide whether a graph is $k$-$(\gamma,\mu)$-choosable. The first algorithm performs an exhaustive enumeration of all possible $(\gamma,\mu)$-list assignments and checks, for each one, whether a valid coloring exists that respects the assigned lists. The second algorithm assumes a known $k$-coloring of the input graph and, based on Theorem~\ref{teo:complexity-lists}, verifies whether this fixed coloring can be adapted to all $(\gamma,\mu)$-lists of size $k$ using the construction derived from Theorem~\ref{teo:general}.

Although both approaches are computationally intensive due to the number of possible list assignments, they guarantee a correct decision for any input graph. Theorem~\ref{teo:complexity-lists} is used in both algorithms to characterize and efficiently generate the space of $(\gamma,\mu)$-lists to be tested.

We define $X$ as the set of all consecutive intervals of size $k$ contained in $\{1,2,\dots,n\}$, with $|X| = n-k+1$. A list assignment is a function $L : V \to X$, where each vertex $v \in V$ is assigned a list $L(v) \in X$. Therefore, the number of possible list assignments is $(n-k+1)^n$:

\begin{theorem}
Let $G = (V,E)$ be a graph with $n = |V|$ vertices, and let $X$ be the set of all consecutive intervals of size $k$ contained in $\{1,2,\dots,n\}$. Then, the number of possible $(\gamma, \mu)$-list assignments, where each vertex $v \in V$ is assigned a list $L(v) \in X$, is exactly $(n-k+1)^n$. Moreover, the set $X$ of all such consecutive intervals can be generated in time $\mathcal{O}(n)$.
\end{theorem}\label{teo:complexity-lists}

\begin{proof}
Each vertex $v \in V$ can be independently assigned any interval from $X$. Since $|X| = n - k + 1$ and each assignment is independent, the total number of possible list assignments is $(n-k+1)^n$, by the basic rule of counting.

To generate the set $X$, observe that each interval of size $k$ corresponds to a consecutive sequence of integers of the form $\{i, i+1, \dots, i+k-1\}$, where $i$ ranges from $1$ to $n-k+1$. Thus, we can generate $X$ by iterating $i$ from $1$ to $n-k+1$, which takes $\mathcal{O}(n)$ time.
\end{proof}

This result directly implies that it is possible to design an exact algorithm for deciding $k$-$(\gamma,\mu)$-choosability. By enumerating all possible $(\gamma,\mu)$-list assignments, and checking for each one whether a proper coloring exists, the algorithm guarantees a correct decision. Although the approach is computationally intensive, it remains exact by construction.

The following Algorithm~\ref{alg:gammaMuChoosability} determines the $(\gamma,\mu)$-choosability number $\chi_{(\gamma,\mu)}(G)$ of a graph $G$, which is the smallest integer $k$ such that $G$ admits a proper coloring from any assignment of intervals of size $k$ to its vertices. This notion relies on $(\gamma,\mu)$-structured lists, where each vertex receives a list consisting of $k$ consecutive integers. Unlike the case where a $k$-coloring is known a priori and can be used to guide the construction of a $(\gamma,\mu)$-coloring, this algorithm assumes no prior knowledge of such a coloring and must exhaustively check all possible $(\gamma,\mu)$-list assignments.

The procedure \texttt{Exists\_List\_Coloring}, presented in \textbf{Algorithm~\ref{alg:existsListColoring}}, is a recursive method designed to determine whether a graph admits at least one valid list coloring. Unlike exhaustive algorithms that enumerate all feasible solutions, this procedure terminates as soon as a single proper coloring is found.

Let $G = (V, E)$ be a simple undirected graph with $n$ vertices. Each vertex $v_i \in V$ is associated with a list $L(v_i)$ of allowable colors, where $L(v_i) \subseteq C$ and $C$ is a finite set of colors. The full list assignment is denoted as $L = \{L(v_1), L(v_2), \dots, L(v_n)\}$. During execution, the algorithm assigns a color to each vertex $v_i$ by setting an attribute \texttt{color} on the vertex object, denoted as $G.v_i.color$.

The algorithm receives as input the graph $G$, the list assignment $L$, a Boolean flag \texttt{exists} initialized as \texttt{FALSE}, and an index $i$ indicating the current vertex under consideration. At each recursive step, the algorithm attempts all available colors in the list $L(v_i)$, backtracking as needed. When all vertices have been assigned a color, the auxiliary procedure \texttt{IsProperColoring($G$)} is invoked to verify that adjacent vertices do not share the same color and that each assigned color belongs to the original list of the corresponding vertex.

\begin{algorithm}[H]
\small
\caption{Exact computation of the $(\gamma,\mu)$-choosability number $\chi_{(\gamma,\mu)}(G)$}
\label{alg:gammaMuChoosability}
\begin{algorithmic}[1]

\State \textbf{Input:} A simple graph $G = (V, E)$.
\State \textbf{Output:} The $(\gamma,\mu)$-choosability number $\chi_{(\gamma,\mu)}(G)$.

\State $n \gets |V|$
\State $k \gets 2$
\While{\textbf{not} \texttt{Is\_K\_GammaMu\_Choosable}($G, k$)}
    \State $k \gets k + 1$
\EndWhile
\State \Return $k$ \Comment{Minimal $k$ such that $G$ is $k$-$(\gamma,\mu)$-choosable}

\vspace{0.5em}
\Function{Is\_K\_GammaMu\_Choosable}{$G = (V, E), k$}
    \State $X \gets \{ \{i, i+1, \dots, i+k-1\} \mid i = 1, \dots, n-k+1 \}$ \Comment{All intervals of size $k$ in $\{1,\dots,n\}$}
    \State $P \gets$ all $n$-tuples of list assignments $L : V \to X$
    \For{\textbf{each} list assignment $L \in P$}
        \If{\textbf{not} \texttt{Exists\_List\_Coloring}($G, L, \texttt{NO}, 1$)}
            \State \Return \texttt{NO}
        \EndIf
    \EndFor
    \State \Return \texttt{YES}
\EndFunction

\end{algorithmic}
\end{algorithm}

The execution halts immediately once a feasible coloring is found, making this algorithm efficient for decision-based scenarios, where the objective is to verify the existence of a solution rather than enumerate them.

\begin{algorithm}[H]
\small
\caption{Exists\_List\_Coloring($G, L, exists, i$)}\label{alg:existsListColoring}
\begin{algorithmic}[1]

\State \textbf{Input:} A simple graph $G = (V, E)$, a list assignment $L = \{L(v_1), \dots, L(v_n)\}$, where each $L(v_i)$ is the set of allowed colors for vertex $v_i$, a Boolean flag \texttt{exists}, and an index $i$ indicating the current vertex being processed.

\State \textbf{Output:} \texttt{TRUE} if there exists a feasible list coloring of $G$; \texttt{FALSE} otherwise.

\If{not $(exists)$}
    \If{$L = \emptyset$}
        \If{\texttt{IsProperColoring}($G$)}
            \State $exists \gets$ \texttt{TRUE};
        \EndIf
    \Else
        \State $l \gets$ remove the first element from list $L$;
        \For{\texttt{color} $\in l$ \textbf{while} $\neg$ exists}
            \State $G.v_i.color \gets$ \texttt{color};
            \State \Call{Exists\_List\_Coloring}{$G, L, exists, i + 1$};
        \EndFor
    \EndIf
\EndIf

\end{algorithmic}
\end{algorithm}

To perform this check, for each value of $k$, the algorithm generates all possible assignments of intervals of size $k$ drawn from the set $\{1, \dots, n\}$, where $n = |V|$. According to Theorem~\ref{teo:complexity-lists}, the number of such assignments is $(n-k+1)^n$. For each list assignment, the algorithm tests whether a proper list coloring exists.

\vspace{1em}
\begin{example}
Let $G$ be a path with $n = 3$ vertices. For $k = 2$, the set $X$ of possible intervals is $\{\{1,2\}, \{2,3\}\}$. There are $2^3 = 8$ total list assignments. The algorithm enumerates all 8 combinations of interval assignments to the vertices, such as:

\begin{itemize}
    \item $L(v_1) = \{1,2\},\; L(v_2) = \{1,2\},\; L(v_3) = \{1,2\}$
    \item $L(v_1) = \{1,2\},\; L(v_2) = \{2,3\},\; L(v_3) = \{2,3\}$
    \item $L(v_1) = \{2,3\},\; L(v_2) = \{1,2\},\; L(v_3) = \{2,3\}$
    \item \dots
\end{itemize}

For each of these, it checks whether a proper coloring can be obtained by selecting one color from each list such that adjacent vertices receive different colors. If all assignments admit a valid coloring, then $G$ is $2$-$(\gamma,\mu)$-choosable.
\end{example}

\vspace{1em}

Before establishing the main complexity result, it is convenient to decompose
the analysis into a sequence of auxiliary lemmas. The algorithms under
consideration involve several distinct combinatorial components—namely, the
construction of all admissible interval lists, the enumeration of all possible
list assignments to the vertices, and the exhaustive recursive exploration of
list colorings. To isolate the contribution of each component to the global
running time, we introduce the following lemmas, each of which captures one of
the fundamental growth mechanisms present in Algorithms~\ref{alg:gammaMuChoosability} and~\ref{alg:existsListColoring}.

\begin{lemma}[Number of list assignments]\label{lemma:1}
Let $G$ have $n$ vertices. For a fixed integer $k$, the procedure
\texttt{Is\_K\_GammaMu\_Choosable} constructs the family
\[
X = \{\{i,i+1,\dots,i+k-1\} : i = 1,\dots,n-k+1\},
\]
and the number of list assignments $L : V \to X$ is
\[
|P| = (n-k+1)^n.
\]
\end{lemma}

\begin{proof}
There are $n-k+1$ intervals of size $k$. Each of the $n$ vertices chooses one.
Thus, $|P| = (n-k+1)^n$. \qedhere
\end{proof}

\begin{lemma}[Running time of Algorithm~\ref{alg:existsListColoring}]\label{lemma:2}
Given lists of size $k$, Algorithm~2 (\texttt{Exists\_List\_Coloring})
performs, in the worst case, $\Theta(k^n)$ recursive steps.
\end{lemma}

\begin{proof}
At recursion depth $i$, the algorithm branches over all $k$ colors in $L(v_i)$.
The recursion depth is $n$, so the total number of explored leaves is $k^n$.
Polynomial overhead per leaf does not change the dominant term. \qedhere
\end{proof}

\begin{lemma}[Cost of \texttt{Is\_K\_GammaMu\_Choosable}]\label{lemma:3}
For fixed $k$, the running time of \texttt{Is\_K\_GammaMu\_Choosable}$(G,k)$
satisfies
\[
T_{\mathrm{IKC}}(n,k) \in \Omega\!\bigl((n-k+1)^n \cdot k^n\bigr).
\]
\end{lemma}

\begin{proof}
By Lemma~\ref{lemma:1}, the procedure tests $(n-k+1)^n$ different list assignments.
By Lemma~\ref{lemma:2}, each call to Algorithm~\ref{alg:existsListColoring} takes $\Theta(k^n)$ time.
Multiplying yields the bound. \qedhere
\end{proof}

The lemmas above provide asymptotically tight lower bounds for the number of
list assignments generated by \texttt{Is\_K\_GammaMu\_Choosable}, as well as for
the recursive search space explored by Algorithm~\ref{alg:existsListColoring}. By combining these results,
we obtain a precise characterization of the worst-case running time of
Algorithm~\ref{alg:gammaMuChoosability}. We are now in a position to state and prove the main theorem,
showing that the overall complexity grows strictly faster than any exponential
function $a^{n}$, and is therefore super-exponential.

\begin{theorem}[Super-exponential complexity of Algorithm~1 and Algorithm~2]
The exact computation of the $(\gamma,\mu)$-choosability number via
Algorithm~1 has worst-case running time
\[
T(n) \in \Omega(n^{2n}),
\]
and is therefore super-exponential in $n$.
\end{theorem}

\begin{proof}
Take $k_0 = \lfloor n/2 \rfloor$ in Lemma 3. Then
$n-k_0+1 \ge n/2$ and $k_0 \ge n/2$, giving
\[
T(n) \;\ge\;
\left(\frac{n}{2}\right)^n \cdot \left(\frac{n}{2}\right)^n
=
\left(\frac{n^2}{4}\right)^n
\in \Omega(n^{2n}).
\]
Since $n^{2n}/a^n \to \infty$ for every constant $a>1$, the algorithm is
super-exponential. \qedhere
\end{proof}

The structure of Algorithm~\ref{alg:gammaMuChoosability} mirrors the logical
form of the $k$-$(\gamma,\mu)$-choosability decision problem. For a fixed
value of $k$, the procedure explicitly enumerates all interval-based list
assignments of size $k$ (line~11), thereby implementing a universal
quantification over list structures. For each such assignment $L$, the
algorithm invokes \texttt{Exists\_List\_Coloring} to determine whether there
exists a proper coloring compatible with $L$.

\section{Conclusion}

This work provides a refined perspective on $k$-$(\gamma,\mu)$-choosability by
investigating list-coloring problems in which the admissible colors assigned to
each vertex are restricted to consecutive integer intervals. By first
formalizing the $k$-$(\gamma,\mu)$-coloring problem and subsequently extending
the analysis to its choosability counterpart, the paper clarifies the structural
and computational effects induced by interval-based list constraints.

As discussed in Section~6, the decision problem associated with
$k$-$(\gamma,\mu)$-choosability is inherently characterized by universal
quantification over all admissible interval list assignments, followed by the
existential search for a compatible coloring. This logical structure places the
problem naturally at the second level of the polynomial hierarchy.
Table~\ref{tab:choosability-structure-complexity} summarizes how this
classification aligns with other choosability generalizations: despite the
restriction to highly structured interval lists, the theoretical complexity of
the general decision problem remains $\Pi_2^P$-complete. The exponential-time
procedures analyzed in this work therefore constitute algorithmic upper bounds
for deciding choosability, rather than indicators of a higher intrinsic
complexity class.

From an algorithmic standpoint, the procedures investigated here provide an
exact mechanism for determining the $(\gamma,\mu)$-choosability number of a
graph. Their computational cost, however, is dominated by the combinatorial
growth induced by the universal enumeration of interval assignments, which
limits their direct applicability to large instances. Nevertheless, these
algorithms serve as a concrete foundation upon which more efficient
implementations may be constructed.

Several promising directions for future work emerge from this analysis. On the
algorithmic side, the procedure \textsc{Is\_K\_GammaMu\_Choosable} could be
substantially strengthened through the incorporation of pruning strategies,
including symmetry breaking among interval assignments, early detection of
infeasible overlap patterns, and constraint-based elimination of redundant
configurations. Such enhancements have the potential to significantly reduce
the effective search space, thereby improving practical tractability and
enabling the evaluation of choosability parameters on graphs of greater size and
structural complexity.

Beyond algorithmic refinements, the results of this study suggest broader
theoretical avenues for exploration. The strong structural correspondences
identified between interval-restricted coloring models motivate the
investigation of analogical or p-analogical relationships, in the sense of
Fellows \emph{et al.}, among other list-based coloring frameworks. In
particular, the following directions appear promising:

\begin{itemize}
    \item \textbf{DP-coloring versus $(\gamma,\mu)$-coloring}: although
    DP-coloring strictly generalizes list coloring, it remains open whether
    restricted correspondence structures may admit analogical transformations
    on specific graph classes.

    \item \textbf{Defective list coloring and interval-based models}: since
    defective colorings relax adjacency constraints whereas $(\gamma,\mu)$-
    colorings restrict the range of admissible colors, studying their
    interaction may reveal new structural phenomena.

    \item \textbf{Weighted or cost-based list coloring}: in graph classes closed
    under pendant extensions, it may be possible to preserve both solution size
    and certificate structure under linear-time reductions analogous to those
    developed for $(\gamma,\mu)$-coloring.

    \item \textbf{Interval coloring and resource-allocation models}: given that
    $(\gamma,\mu)$-coloring inherently captures bounded and ordered intervals,
    potential connections with classical interval coloring and scheduling
    problems merit further investigation.
\end{itemize}

Overall, this work advances the understanding of how list-structure restrictions
influence the computational landscape of graph coloring problems. By
disentangling algorithmic upper bounds from theoretical complexity
classifications and by clarifying the role of interval constraints in
choosability, the paper establishes a solid foundation for further algorithmic
and structural investigations into constrained coloring frameworks.

\bibliographystyle{apalike-sol}
\bibliography{refs}

@article{gama2018choosability,
  title={Choosability in bounded sequential list coloring},
  author={Gama, Simone and de Freitas, Rosiane and Salvatierra, M{\'a}rio},
  journal={arXiv preprint arXiv:1812.11685},
  year={2018}
}

@article{bernshteyn2017dp,
  title={On DP-coloring of graphs and multigraphs},
  author={Bernshteyn, A Yu and Kostochka, AV and Pron, SP},
  journal={Siberian Mathematical Journal},
  volume={58},
  number={1},
  pages={28--36},
  year={2017},
  publisher={Springer}
}

@article{havet2009improper,
  title={Improper coloring of unit disk graphs},
  author={Havet, Fr{\'e}d{\'e}ric and Kang, Ross J and Sereni, Jean-S{\'e}bastien},
  journal={Networks: An International Journal},
  volume={54},
  number={3},
  pages={150--164},
  year={2009},
  publisher={Wiley Online Library}
}

@article{zhu2009line,
  title={On-line list colouring of graphs},
  author={Zhu, Xuding},
  journal={the electronic journal of combinatorics},
  pages={R127--R127},
  year={2009}
}

@article{schauz2009mr,
  title={Mr. paint and mrs. correct},
  author={Schauz, Uwe},
  journal={the electronic journal of combinatorics},
  pages={R77--R77},
  year={2009}
}

@article{dvorak2015correspondence,
  title={Correspondence coloring and its application to list-coloring planar graphs without cycles of lengths 4 to 8},
  author={Dvořák, Zdenek and Postle, Luke},
  journal={arXiv preprint arXiv:1508.03437},
  year={2015}
}

@article{cowen1997defective,
  title={Defective coloring revisited},
  author={Cowen, Lenore and Goddard, Wayne and Jesurum, C Esther},
  journal={Journal of Graph Theory},
  volume={24},
  number={3},
  pages={205--219},
  year={1997},
  publisher={Wiley Online Library}
}

@article{gutner2009some,
  title={Some results on (a: b)-choosability},
  author={Gutner, Shai and Tarsi, Michael},
  journal={Discrete Mathematics},
  volume={309},
  number={8},
  pages={2260--2270},
  year={2009},
  publisher={Elsevier}
}

@article{gama2019aspects,
  title={Aspects of the complexity of ($\gamma$, $\mu$)-coloring},
  author={Gama, Simone and de Freitas, Rosiane and Souza, U{\'e}verton S},
  journal={Matem{\'a}tica Contempor{\^a}nea},
  volume={46},
  pages={248--255},
  year={2019}
}

@article{ahn2022towards,
  title={Towards constant-factor approximation for chordal/distance-hereditary vertex deletion},
  author={Ahn, Jungho and Kim, Eun Jung and Lee, Euiwoong},
  journal={Algorithmica},
  volume={84},
  number={7},
  pages={2106--2133},
  year={2022},
  publisher={Springer}
}

@book{chartrand2019chromatic,
  title={Chromatic graph theory},
  author={Chartrand, Gary and Zhang, Ping},
  year={2019},
  publisher={Chapman and Hall/CRC}
}

@article{fellows2011complexity,
	title={On the complexity of some colorful problems parameterized by treewidth},
	author={Fellows, Michael R and Fomin, Fedor V and Lokshtanov, Daniel and Rosamond, Frances and Saurabh, Saket and Szeider, Stefan and Thomassen, Carsten},
	journal={Information and Computation},
	volume={209},
	number={2},
	pages={143--153},
	year={2011},
	doi = "10.1016/j.ic.2010.11.026",
	publisher={Elsevier},
}

@article{fiala2011parameterized,
	title={Parameterized complexity of coloring problems: Treewidth versus vertex cover},
	author={Fiala, Ji{\v{r}}{\'\i} and Golovach, Petr A and Kratochv{\'\i}l, Jan},
	journal={Theoretical Computer Science},
	volume={412},
	number={23},
	pages={2513--2523},
	year={2011},
	publisher={Elsevier}
}

@article{couturier2012parameterized,
	title={On the parameterized complexity of coloring graphs in the absence of a linear forest},
	author={Couturier, Jean-Fran{\c{c}}ois and Golovach, Petr A and Kratsch, Dieter and Paulusma, Dani{\"e}l},
	journal={Journal of Discrete Algorithms},
	volume={15},
	pages={56--62},
	year={2012},
	publisher={Elsevier}
}

@article{golovach2017survey,
	title={A survey on the computational complexity of coloring graphs with forbidden subgraphs},
	author={Golovach, Petr A and Johnson, Matthew and Paulusma, Dani{\"e}l and Song, Jian},
	journal={Journal of Graph Theory},
	volume={84},
	number={4},
	pages={331--363},
	year={2017},
	publisher={Wiley Online Library}
}

@article{hofmannproofs,
  title={Proofs from THE BOOK},
  author={Hofmann, Karl H},
  publisher={Citeseer},
  journal={10º Conference on Combinatorial, Graph Theory, and Computing},
  year={1979}
}

@article{fellows2007fixed,
	title={On the fixed-parameter intractability and tractability of multiple-interval graph problems},
	author={Fellows, Michael R and Hermelin, Danny and Rosamond, Frances},
	journal={Unpublished Result},
	year={2007}
}

@article{downey1998parameterized,
	title={Parameterized circuit complexity and the W hierarchy},
	author={Downey, Rodney G and Fellows, Michael R and Regan, Kenneth W},
	journal={Theoretical Computer Science},
	volume={191},
	number={1-2},
	pages={97--115},
	year={1998},
	publisher={Elsevier}
}

@inproceedings{abrahamson1989complexity,
	title={On the complexity of fixed parameter problems},
	author={Abrahamson, Karl R and Fellows, MR and Ellis, JA and Mata, Manuel E},
	booktitle={Foundations of Computer Science, 1989., 30th Annual Symposium on},
	pages={210--215},
	year={1989},
	organization={IEEE}
}

@article{cai1997fixed,
	title={On fixed-parameter tractability and approximability of NP optimization problems},
	author={Cai, Liming and Chen, Jianer},
	journal={Journal of Computer and System Sciences},
	volume={54},
	number={3},
	pages={465--474},
	year={1997},
	publisher={Elsevier}
}

@inproceedings{Downey1999ParameterizedC,
	title={Parameterized Complexity},
	author={Rodney G. Downey and Michael R. Fellows},
	booktitle={Monographs in Computer Science},
	year={1999}
}

@article{niedermeier2002invitation,
	title={Invitation to fixed-parameter algorithms},
	author={Niedermeier, Rolf},
	journal={Habilitationschrift, University of T{\"u}bingen},
	volume={19},
	year={2002},
	publisher={Citeseer}
}

@article{bonomo2005between,
	author = {Flavia Bonomo and Mariano Cecowski},
	title = {Between coloring and list-coloring: $\mu$-coloring},
	journal = {{Electronic Notes in Discrete Mathematics}},
	year = {2005},
	volume = {19},
	pages = {117--123}
}

@article{bonomo2009exploring,
	title={Exploring the complexity boundary between coloring and list-coloring},
	author={Bonomo, Flavia and Dur{\'a}n, Guillermo and Marenco, Javier},
	journal={Annals of Operations Research},
	volume={169},
	number={1},
	pages={3--16},
	year={2009},
	publisher={Springer}
}

@article{bonomo2012coloring,
	title={On coloring problems with local constraints},
	author={Bonomo, Flavia and Faenza, Yuri and Oriolo, Gianpaolo},
	journal={Discrete Mathematics},
	volume={312},
	number={12-13},
	pages={2027--2039},
	year={2012},
	publisher={Elsevier}
}

@phdthesis{gravier1996coloration,
  title={Coloration et produits de graphes},
  author={Gravier, Sylvain},
  year={1996},
  school={Universit{\'e} Joseph Fourier (Grenoble; 1971-2015)}
}

@article{arnborg1989linear,
	title={Linear time algorithms for NP-hard problems restricted to partial k-trees},
	author={Arnborg, Stefan and Proskurowski, Andrzej},
	journal={Discrete applied mathematics},
	volume={23},
	number={1},
	pages={11--24},
	year={1989},
	publisher={North-Holland}
}

@article{biro1992precoloring,
	title={Precoloring extension. I. Interval graphs},
	author={Biro, Mikl{\'o}s and Hujter, Mih{\'a}ly and Tuza, Zs},
	journal={Discrete Mathematics},
	volume={100},
	number={1-3},
	pages={267--279},
	year={1992},
	publisher={Elsevier}
}

@phdthesis{song2013graph,
	title={Graph Colouring with Input Restrictions},
	author={Song, Jian},
	year={2013},
	school={Durham University}
}

@inproceedings{jansen1997optimum,
	title={The optimum cost chromatic partition problem},
	author={Jansen, Klaus},
	booktitle={Italian Conference on Algorithms and Complexity},
	pages={25--36},
	year={1997},
	organization={Springer}
}

@article{kubale1992some,
	title={Some results concerning the complexity of restricted colorings of graphs},
	author={Kubale, Marek},
	journal={Discrete Applied Mathematics},
	volume={36},
	number={1},
	pages={35--46},
	year={1992},
	publisher={North-Holland}
}

@article{jansen1997generalized,
	title={Generalized coloring for tree-like graphs},
	author={Jansen, Klaus and Scheffler, Petra},
	journal={Discrete Applied Mathematics},
	volume={75},
	number={2},
	pages={135--155},
	year={1997},
	publisher={Elsevier}
}

@article{voigt1993list,
	title={List colourings of planar graphs.},
	author={Voigt, Margit},
	journal={Discrete Mathematics},
	volume={120},
	number={1-3},
	pages={215--219},
	year={1993}
}

@article{alon2000degrees,
	title={Degrees and choice numbers},
	author={Alon, Noga},
	journal={Random Structures \& Algorithms},
	volume={16},
	number={4},
	pages={364--368},
	year={2000},
	publisher={Wiley Online Library}
}

@article{alon1993restricted,
	title={Restricted colorings of graphs},
	author={Alon, Noga},
	journal={Surveys in combinatorics},
	volume={187},
	pages={1--33},
	year={1993}
}

@article{lastrina2012list,
	title={List-coloring and sum-list-coloring problems on graphs},
	author={Lastrina, Michelle Anne},
        journal={ProQuest LLC},
	year={2012},
	publisher={Digital Repository@ Iowa State University}
}

@phdthesis{dinitz1980lower,
	title={New Lower bounds for the number of pairwise orthogonal symmetric Latin squares},
	author={Dinitz, Jeffrey H},
	journal={Proceedings of the Tenth Southeastern Conference on Combinatorics,
	Graph Theory and Computing (Florida Atlantic Univ., Boca Raton, Fla., 1979)},
	year={1980},
	pages={393--398},
	school={The Ohio State University}
}

@article{mihokchromatic,
	title={Chromatic number of classes of graphs with prescribed cycle lengths},
	author={Mih{\'o}k, P and Schiermeyer, I},
        year={1997},
	journal={submitted for publication}
}

@article{randerath2004vertex,
	title={Vertex colouring and forbidden subgraphs--a survey},
	author={Randerath, Bert and Schiermeyer, Ingo},
	journal={Graphs and Combinatorics},
	volume={20},
	number={1},
	pages={1--40},
	year={2004},
	publisher={Springer}
}

@incollection{karp1972reducibility,
	title={Reducibility among combinatorial problems},
	author={Karp, Richard M},
	booktitle={Complexity of computer computations},
	pages={85--103},
	year={1972},
	publisher={Springer}
}

@article{vizing1976coloring,
	title={Coloring the vertices of a graph in prescribed colors},
	author={Vizing, Vadim G},
	journal={Diskret. Analiz},
	volume={29},
	number={3},
	pages={10},
	year={1976}
}

@article{Erdos1959GraphProbab,
	title={Graph theory and probability},
	author={Paul Erd{\"{o}}s},
	journal={Canad. J. Math.},
	pages={34--38},
	year={1959},
	doi = "https://doi.org/10.4153/CJM-1959-003-9",
	publisher={CJM}
}

@article{lovasz1979shannon,
	title={On the Shannon capacity of a graph},
	author={Lov{\'a}sz, L{\'a}szl{\'o}},
	journal={IEEE Transactions on Information theory},
	volume={25},
	number={1},
	pages={1--7},
	year={1979},
	publisher={IEEE}
}

@incollection{pardalos1998graph,
	title={The graph coloring problem: A bibliographic survey},
	author={Pardalos, Panos M and Mavridou, Thelma and Xue, Jue},
	booktitle={Handbook of combinatorial optimization},
	pages={1077--1141},
	year={1998},
	publisher={Springer}
}

@article{lovasz1975three,
	title={Three short proofs in graph theory},
	author={Lov{\'a}sz, L{\'a}szl{\'o}},
	journal={Journal of Combinatorial Theory, Series B},
	volume={19},
	number={3},
	pages={269--271},
	year={1975},
	publisher={Academic Press}
}

@inproceedings{brooks1941colouring,
	title={On colouring the nodes of a network},
	author={Brooks, Rowland Leonard},
	booktitle={Mathematical Proceedings of the Cambridge Philosophical Society},
	volume={37},
	pages={194--197},
	year={1941},
	organization={Cambridge University Press}
}

@article{korman1979graph,
	title={The graph-colouring problem},
	author={Korman, Samuel M},
	journal={Combinatorial optimization},
	pages={211--235},
	year={1979},
	publisher={Wiley, New York}
}

@article{fellows2015tractability,
	title	= {Tractability and hardness of flood-filling games on trees},
	author	= {Fellows, Michael R and dos Santos Souza, U{\'e}verton and Protti, F{\'a}bio and da Silva, Maise Dantas},
	journal	= {Theoretical Computer Science},
	volume	= {576},
	pages	= {102--116},
	year	= {2015},
	publisher={Elsevier}
}

@article{erdos:1979,
	author  = {Paul Erd{\"{o}}s and L. Rubin},
	title   = {Chosability in graphs},
	journal = {Proceedings West Coast Conference on Combinatorics},
	year    = {1979}
}

@article{thomassen1994every,
	title={Every planar graph is 5-choosable},
	author={Thomassen, Carsten},
	journal={Journal of Combinatorial Theory Series B},
	volume={62},
	number={1},
	pages={180--181},
	year={1994},
	publisher={Academic Press, Inc.}
}

\end{document}